\PassOptionsToPackage{dvipsnames}{xcolor}
\documentclass[sigplan,screen,pbalance]{acmart}
\copyrightyear{2027}
\acmYear{2027}
\setcopyright{cc}
\setcctype{by}
\acmConference[ASPLOS '27]{Proceedings of the 32nd ACM International Conference on Architectural Support for Programming Languages and Operating Systems, Volume 1}{April 11--15, 2027}{Heraklion, Greece.}
\acmBooktitle{Proceedings of the 32nd ACM International Conference on Architectural Support for Programming Languages and Operating Systems, Volume 1 (ASPLOS '27), April 11--15, 2027, Heraklion, Greece}
\acmISBN{979-8-4007-3012-2/2027/04}
\acmDOI{10.1145/3845814.3850041}


\usepackage{listings}
\usepackage{adjustbox}
\usepackage[inline]{enumitem}
\usepackage{subcaption}
\usepackage{keytheorems}
\usepackage{tikz}
\usepackage{xspace}
\usepackage[group-separator={,},group-minimum-digits=4]{siunitx}
\usetikzlibrary{arrows.meta, positioning, calc, fit, backgrounds}
\usepackage[capitalise,nameinlink]{cleveref}
\crefname{section}{}{}
\Crefname{section}{}{}
\creflabelformat{section}{\S#2#1#3}
\crefname{appendix}{Appendix}{Appendices}
\Crefname{appendix}{Appendix}{Appendices}
\crefname{supplement}{Supplement}{Supplements}
\Crefname{supplement}{Supplement}{Supplements}

\makeatletter
\let\section@prefix\@empty
\def\@seccntformat#1{%
  \ifcsname #1@prefix\endcsname\csname #1@prefix\endcsname\fi
  \csname the#1\endcsname\quad}
\newcommand{\startappendix}{%
  \appendix
  \def\section@prefix{Appendix~}%
  \crefalias{section}{appendix}%
  \crefalias{subsection}{appendix}%
  \crefalias{subsubsection}{appendix}}
\newcommand{\startsupplement}{%
  \setcounter{section}{0}%
  \def\section@prefix{Supplement~}%
  \renewcommand{\theHsection}{supplement.\Alph{section}}%
  \crefalias{section}{supplement}%
  \crefalias{subsection}{supplement}%
  \crefalias{subsubsection}{supplement}}
\makeatother

\usepackage{newfloat}
\DeclareFloatingEnvironment[name=Listing]{listing}
\newkeytheorem{theorem}[style=definition]
\newkeytheorem{lemma}[style=definition]
\newkeytheorem{definition}[style=acmdefinition, sibling=theorem]

\newtheoremstyle{constraint}%
  {\topsep}{\topsep}{}{}{\bfseries}{}{ }%
  {(C\thmnumber{#2})\thmnote{ #3}}
\theoremstyle{constraint}
\newtheorem{constraint}{}
\newcommand{\cmark}[1]{\textbf{(C\ref{#1})}}
\newcommand{\LiveIn}{\operatorname{LiveIn}}

\newcommand{\eligibles}{\mathcal{V}_{\mathrm{cand}}}
\newcommand{\ineligibles}{\mathcal{V}_{\mathrm{fixed}}}
\newcommand{\variables}{\mathcal{V}}
\newcommand{\defs}{\mathcal{D}}

\newcommand{\ebuffer}{E_{\mathrm{budget}}}
\newcommand{\eaccum}{e}

\newcommand{\deltaE}{\delta}
\newcommand{\deltaEread}{\delta_r}
\newcommand{\deltaEwrite}{\delta_w}

\newcommand{\Eread}{\rho_r}
\newcommand{\Ewrite}{\rho_w}

\newcommand{\Epro}{E_{\mathrm{pro}}}
\newcommand{\Eepi}{E_{\mathrm{epi}}}

\newcommand{\Erst}{E^{r}}
\newcommand{\Esv}{E^{s}}

\newcommand{\Ebase}{E^{0}}

\newcommand{\eblk}{\varepsilon}
\newcommand{\estart}{\varepsilon^{+}}
\newcommand{\eend}{\varepsilon^{-}}

\newcommand{\Svm}{S_{\mathrm{VM}}}

\newcommand{\size}{\operatorname{size}}
\newcommand{\numReads}{\operatorname{numReads}}
\newcommand{\numWrites}{\operatorname{numWrites}}

\newcommand{\pred}{\mathrm{pred}}
\renewcommand{\succ}{\mathrm{succ}}

\newcommand{\tool}{\textsc{Bao}\xspace}
\newcommand{\rockclimb}{\textsc{RockClimb}\xspace}
\newcommand{\schematic}{\textsc{Schematic}\xspace}
\newcommand{\schematicO}{\textsc{Schematic-O3}\xspace}

\title[\tool]{\tool: Automatic Region Placement and Memory Allocation for Intermittent Computing}

\author{Byeongjee Kang}
\orcid{0000-0002-2817-5123}
\affiliation{%
  \institution{Carnegie Mellon University}
  \city{Pittsburgh}
  \state{PA}
  \country{USA}}
\email{byeongjee@cmu.edu}

\author{Jan Hoffmann}
\orcid{0000-0001-8326-0788}
\affiliation{%
  \institution{Carnegie Mellon University}
  \city{Pittsburgh}
  \state{PA}
  \country{USA}}
\email{jhoffmann@cmu.edu}

\author{Limin Jia}
\orcid{0000-0002-8160-349X}
\affiliation{%
  \institution{Carnegie Mellon University}
  \city{Pittsburgh}
  \state{PA}
  \country{USA}}
\email{liminjia@andrew.cmu.edu}

\author{Brandon Lucia}
\orcid{0000-0002-8902-2518}
\affiliation{%
  \institution{Carnegie Mellon University}
  \city{Pittsburgh}
  \state{PA}
  \country{USA}}
\affiliation{%
  \institution{Efficient Computer}
  \city{Pittsburgh}
  \state{PA}
  \country{USA}}
\email{blucia@andrew.cmu.edu}

\author{Feras A.~Saad}
\orcid{0000-0002-0505-795X}
\affiliation{%
  \institution{Carnegie Mellon University}
  \city{Pittsburgh}
  \state{PA}
  \country{USA}}
\email{fsaad@cmu.edu}

\begin{document}

\begin{abstract}
Intermittent computing enables batteryless embedded devices to operate
in harsh environment,
but frequent power failures interrupt program
execution and require careful management of state across power cycles.
The core challenge is to guarantee both memory consistency and forward progress
while minimizing the energy overhead of checkpointing.
Recent compile-time approaches identify regions that fit in the energy buffer.
At runtime, the device waits and recharges between
the regions.
However, they still rely on greedy or path-local heuristics
that commit to local decisions and can miss globally lower-overhead boundary
placements.

To address this limitation, we present \tool,
a system that jointly finds optimal energy-aware region formation and memory
allocation decisions with formal correctness guarantees, by
formulating it
as a mixed-integer linear program.
We prove that any feasible solution guarantees energy safety
and implement our approach in LLVM.
Our evaluation on 13 benchmarks across 3 capacitor sizes shows that
\tool outperforms existing baselines,
achieving 10\% faster execution and 52\%
fewer region boundary hits on average compared to the best baseline.
\end{abstract}

\maketitle

\section{Introduction}
\label{sec:introduction}

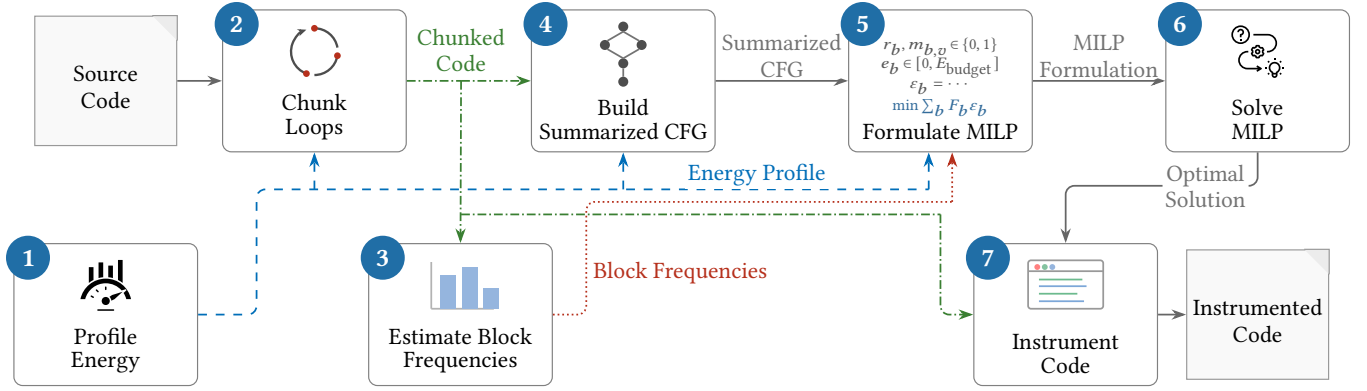
\begin{figure*}[ht]
\centering
\begin{adjustbox}{width=\textwidth}
\begin{tikzpicture}[
  font=\small,
  stepbox/.style={
    rectangle, rounded corners=3pt, draw=black!50, fill=white,
    minimum width=2.2cm, minimum height=1.7cm, inner sep=2pt
  },
  badge/.style={
    circle, fill=RoyalBlue!80!black, text=white,
    font=\small\bfseries, inner sep=0pt, minimum size=15pt
  },
  arr/.style={-{Stealth[length=5pt,width=3.5pt]}, semithick, black!60, rounded corners=3pt},
  lbl/.style={font=\footnotesize, text=black!55, fill=white, inner sep=1pt},
  energy/.style={semithick, NavyBlue, dashed, rounded corners=3pt},
  freq/.style={semithick, BrickRed, densely dotted, rounded corners=3pt},
  chunked/.style={semithick, OliveGreen, densely dash dot, rounded corners=3pt},
]

\node[minimum width=1.7cm, minimum height=1.55cm,
      draw=black!40, fill=gray!5, inner sep=4pt] (src) at (0,0) {};
\fill[gray!15] ([shift={(-0.3,0)}]src.north east) -- (src.north east)
    -- ([shift={(0,-0.3)}]src.north east) -- cycle;
\draw[black!40] ([shift={(-0.3,0)}]src.north east) -- (src.north east)
    -- ([shift={(0,-0.3)}]src.north east);
\node[font=\footnotesize, align=center] at ([yshift=-2pt]src.center) {Source\\[-1pt]Code};

\node[stepbox] (s2) at (2.5, 0) {};
\node[badge] at ([shift={(5pt,-5pt)}]s2.north west) {2};
\begin{scope}[shift={($(s2.center)+(0,0.35)$)}]
  \draw[thick, black!70] (90:0.28) arc(90:10:0.28);
  \draw[thick, black!70] (-10:0.28) arc(-10:-90:0.28);
  \draw[thick, black!70, -{Stealth[length=3pt,width=2.5pt]}]
    (-110:0.28) arc(-110:-250:0.28);
  \foreach \a in {90,-10,-110} {
    \fill[BrickRed] (\a:0.28) circle(0.04);
  }
\end{scope}
\node[font=\footnotesize, anchor=south, align=center] at ([yshift=1pt]s2.south)
    {Chunk\\[-2pt]Loops};

\node[stepbox] (s3) at (6.2, 0) {};
\node[badge] at ([shift={(5pt,-5pt)}]s3.north west) {4};
\begin{scope}[shift={($(s3.center)+(0,0.35)$)}]
  \fill[black!70] (0,0.30) circle(0.06);
  \fill[black!70] (-0.22,0.08) circle(0.06);
  \fill[black!70] (0.22,0.08) circle(0.06);
  \fill[black!70] (0,-0.14) circle(0.06);
  \fill[black!70] (0,-0.36) circle(0.06);
  \draw[semithick, black!55] (0,0.24) -- (-0.19,0.13);
  \draw[semithick, black!55] (0,0.24) -- (0.19,0.13);
  \draw[semithick, black!55] (-0.19,0.03) -- (0,-0.08);
  \draw[semithick, black!55] (0.19,0.03) -- (0,-0.08);
  \draw[semithick, black!55] (0,-0.20) -- (0,-0.30);
\end{scope}
\node[font=\footnotesize, anchor=south, align=center] at ([yshift=1pt]s3.south)
    {Build\\[-2pt]Summarized CFG};

\node[stepbox] (s4) at (10.0, 0) {};
\node[badge] at ([shift={(5pt,-5pt)}]s4.north west) {5};
\begin{scope}[shift={($(s4.center)+(0,0.1)$)}, font=\tiny]
  \node[text=black!65, font=\fontsize{5}{6.5}\selectfont] at (0,0.30) {$r_b, m_{b,v} \!\in\! \{0,1\}$};
  \node[text=black!65, font=\fontsize{5}{6.5}\selectfont] at (0,0.06) {$e_b \!\in\! [0, \ebuffer]$};
  \node[text=black!65, font=\fontsize{5}{6.5}\selectfont] at (0,-0.18) {$\varepsilon_b = \cdots$};
  \node[text=RoyalBlue!70!black, font=\fontsize{5}{6.5}\selectfont] at (0,-0.42) {$\min \textstyle\sum_b F_b \varepsilon_b$};
\end{scope}
\node[font=\footnotesize, anchor=south, align=center] at ([yshift=1pt]s4.south)
    {Formulate MILP};

\node[stepbox] (s5) at (13.8, 0) {};
\node[badge] at ([shift={(5pt,-5pt)}]s5.north west) {6};
\node at ([yshift=0.35cm]s5.center) {\includegraphics[height=0.65cm]{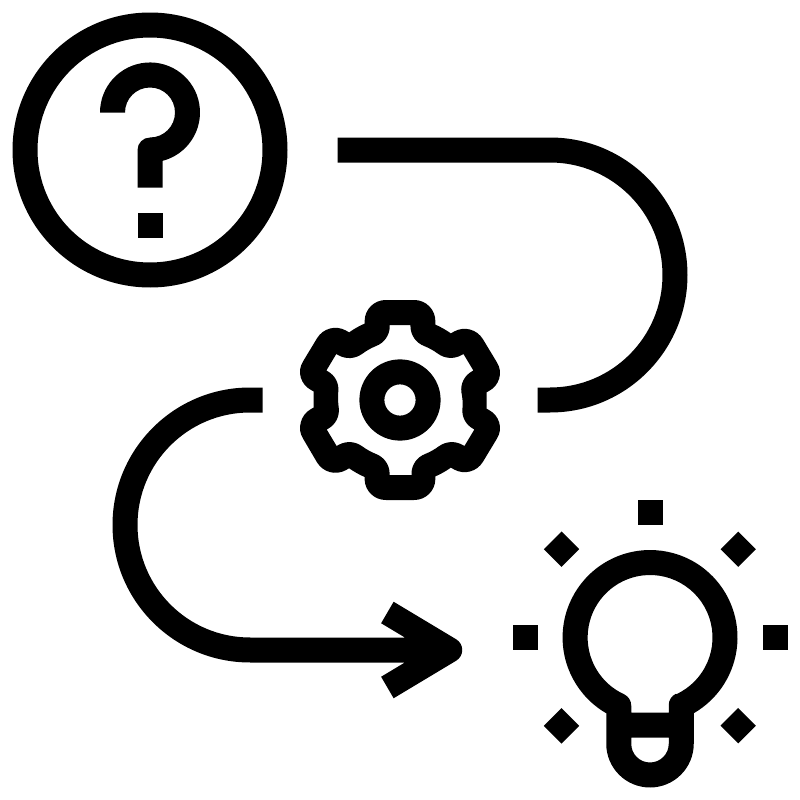}};
\node[font=\footnotesize, anchor=south, align=center] at ([yshift=1pt]s5.south)
    {Solve\\[-2pt]MILP};

\node[stepbox] (s1) at (0, -2.8) {};
\node[badge] at ([shift={(5pt,-5pt)}]s1.north west) {1};
\node at ([yshift=0.35cm]s1.center) {\includegraphics[height=0.65cm]{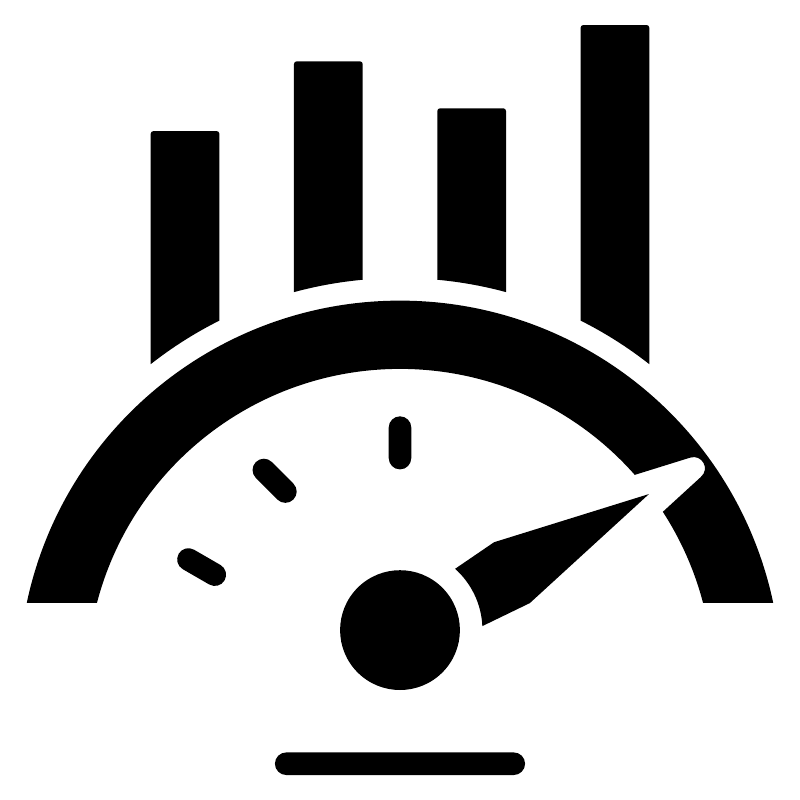}};
\node[font=\footnotesize, anchor=south, align=center] at ([yshift=1pt]s1.south)
    {Profile\\[-2pt]Energy};

\node[stepbox] (sfreq) at (4.25, -2.8) {};
\node[badge] at ([shift={(5pt,-5pt)}]sfreq.north west) {3};
\begin{scope}[shift={($(sfreq.center)+(0,0.35)$)}]
  \draw[semithick, black!30] (-0.35,-0.28) -- (-0.35,0.28);
  \draw[semithick, black!30] (-0.35,-0.28) -- (0.40,-0.28);
  \fill[RoyalBlue!40] (-0.24,-0.28) rectangle (-0.06,0.12);
  \fill[RoyalBlue!40] (0.02,-0.28) rectangle (0.20,0.22);
  \fill[RoyalBlue!40] (0.28,-0.28) rectangle (0.46,-0.04);
\end{scope}
\node[font=\footnotesize, anchor=south, align=center] at ([yshift=1pt]sfreq.south)
    {Estimate Block\\[-2pt]Frequencies};

\node[stepbox] (s6) at (11.5, -2.8) {};
\node[badge] at ([shift={(5pt,-5pt)}]s6.north west) {7};
\begin{scope}[shift={($(s6.center)+(0,0.35)$)}]
  \draw[semithick, black!40, rounded corners=1pt] (-0.45,-0.28) rectangle (0.45,0.30);
  \draw[semithick, black!40] (-0.45,0.17) -- (0.45,0.17);
  \fill[red!50] (-0.34,0.22) circle(0.035);
  \fill[ForestGreen!50] (-0.25,0.22) circle(0.035);
  \fill[RoyalBlue!40] (-0.16,0.22) circle(0.035);
  \draw[semithick, ForestGreen!60] (-0.32,0.07) -- (0.20,0.07);
  \draw[semithick, RoyalBlue!50] (-0.32,-0.01) -- (0.25,-0.01);
  \draw[semithick, ForestGreen!60] (-0.32,-0.10) -- (0.14,-0.10);
  \draw[semithick, RoyalBlue!50] (-0.32,-0.19) -- (0.22,-0.19);
\end{scope}
\node[font=\footnotesize, anchor=south, align=center] at ([yshift=1pt]s6.south)
    {Instrument\\[-2pt]Code};

\node[minimum width=1.7cm, minimum height=1.55cm,
      draw=black!40, fill=gray!5, inner sep=4pt] (out) at (13.8,-2.8) {};
\fill[gray!15] ([shift={(-0.3,0)}]out.north east) -- (out.north east)
    -- ([shift={(0,-0.3)}]out.north east) -- cycle;
\draw[black!40] ([shift={(-0.3,0)}]out.north east) -- (out.north east)
    -- ([shift={(0,-0.3)}]out.north east);
\node[font=\footnotesize, align=center] at ([yshift=-2pt]out.center) {Instrumented\\[-1pt]Code};

\draw[arr] (src.east) -- (s2.west);

\coordinate (ccfork) at ($(s2.east)+(0.65,0)$);
\draw[chunked] (s2.east) -- (ccfork);
\node[lbl, above=2pt, text=OliveGreen, align=center] at (ccfork) {Chunked\\[-1pt]Code};
\draw[arr, OliveGreen, densely dash dot] (ccfork) -- (s3.west);

\draw[arr] (s3.east) -- node[lbl, above, align=center] {Summarized\\[-1pt]CFG} (s4.west);

\draw[arr] (s4.east) -- node[lbl, above, align=center] {MILP\\[-1pt]Formulation} (s5.west);

\draw[arr] (s5.south) -- ++(0,-0.4) -| node[lbl, right, align=center, pos=0.25] {Optimal\\[-1pt]Solution} (s6.north);

\draw[arr] (s6.east) -- (out.west);

\draw[energy] (s1.east) -- ++(0.7,0) -- ++(0,1.5)
    -- ([xshift=-4pt]s4.south |- 0,-1.3);
\node[lbl, above, text=NavyBlue] at (7.8,-1.3) {Energy Profile};
\draw[arr, NavyBlue, dashed] (2.5,-1.3) -- (s2.south);
\draw[arr, NavyBlue, dashed] (6.2,-1.3) -- (s3.south);
\draw[arr, NavyBlue, dashed] ([xshift=-4pt]s4.south |- 0,-1.3) -- ([xshift=-4pt]s4.south);

\draw[freq] (sfreq.east) -- ++(0.4,0) -- ++(0,1.35)
    -- ([xshift=4pt]s4.south |- 0,-1.45);
\node[lbl, anchor=west, text=BrickRed] at (5.8,-2.3) {Block Frequencies};
\draw[arr, BrickRed, densely dotted] ([xshift=4pt]s4.south |- 0,-1.45) -- ([xshift=4pt]s4.south);

\draw[arr, OliveGreen, densely dash dot] (ccfork) -- (sfreq.north);
\draw[chunked] (ccfork |- 0,-1.6) -- (10.0,-1.6);
\draw[arr, OliveGreen, densely dash dot] (10.0,-1.6) -- (10.0,-2.8) -- (s6.west);

\end{tikzpicture}
\end{adjustbox}
\captionsetup{skip=10pt}
\caption{Overview of \tool's pipeline.
  Steps~1 and~3 are detailed in \cref{sec:parameter-estimation};
  Steps~2 and~4--7 in \cref{sec:algorithm}.}
\Description{Pipeline diagram showing the steps of automatic region-boundary placement: energy profiling, loop chunking, frequency estimation, summarized CFG construction, MILP formulation, MILP solving, and code instrumentation.}
\label{fig:overview}
\end{figure*}

Intermittent computing enables batteryless embedded devices to operate in
environments where replacing or recharging batteries is impractical, such as
implanted medical sensors, satellites, and remote environmental monitors.
These devices harvest energy from ambient sources
and store it in small capacitors.
Due to limited stored energy, program execution is
frequently interrupted by power failures.
It is challenging to ensure
correctness and energy efficiency in the presence of frequent power
failures and reboots~\cite{lucia2017intermittent}.

To guarantee \emph{forward progress}, prior
work has developed intermittent system support to save execution
state in checkpoints and restore the state upon
reboot from a power failure~\cite{cleancut, samoyed,
mementos, hibernus, hibernuspp, maeng2018adaptive, coati,
van2016intermittent, auto-tune, alpaca, chain, alpaca, lucia2015simpler}.
The content and placement of these checkpoints impact both correctness
and energy efficiency of the intermittent execution.  Checkpointing
the entire volatile memory (VM) and non-volatile memory (NVM) ensures correctness, but
consumes significant energy; while checkpointing only the VM state
may introduce memory consistency bugs~\cite{lucia2015simpler, formal-foundation}.
Frequent checkpointing helps ensure forward progress,
but the energy overhead of checkpointing and state restoration can
consume a substantial portion of a device’s limited energy
budget. Conversely, infrequent checkpointing risks exhausting
available energy before reaching the next checkpoint, undermining
progress.

Prior work investigates statically-placed
checkpoints~\cite{lucia2015simpler,van2016intermittent, auto-tune, alpaca, chain, cleancut, tiling,wario,lact} and
voltage-triggered, dynamically
placed checkpointing systems~\cite{mementos,hibernus,hibernuspp, harvos, samoyed, nvp-design, maeng2018adaptive}.
Some systems roll back to the most recently checkpointed state upon
reboot~\cite{van2016intermittent, alpaca, lucia2015simpler, maeng2018adaptive,
chain}; while others save the state right before power-off and are
thus \emph{rollback-free}~\cite{hibernus, hibernuspp, samoyed}.
To address both memory consistency and forward progress,
recent rollback-free systems, \rockclimb~\cite{rockclimb} and \schematic~\cite{schematic2024}, divide a program
into \emph{energy-aware} regions at compile time, each of which is guaranteed to complete within a
single charge cycle.
The device recharges between them, and state is saved and restored only at
region boundaries.
Their energy efficiency
depends critically on where region boundaries are placed, which are decided using heuristics.
\rockclimb~\cite{rockclimb} forms regions through a greedy traversal of the
control-flow graph.
\schematic~\cite{schematic2024} co-optimizes region boundaries with memory
allocation (i.e., whether each variable resides in VM or NVM), since the allocation affects both the cost of ordinary execution and the cost of
saving and restoring state at each boundary.
Both, however, commit to local decisions rather than optimizing
over the entire control-flow graph, which can miss lower-overhead
placements from a global view, or insert unnecessary boundaries.

Intermittent computing is a compile-once, run-many setting: batteryless
devices are deployed for long-lived operation in environments where updates
are difficult or impossible, so additional compile-time effort is worthwhile
when it lowers steady-state overhead and strengthens forward progress guarantees.
Moreover, post-deployment debugging is infeasible, which raises the bar for
correctness: the compiler should produce binaries
that are guaranteed to be correct under all possible execution scenarios.
These considerations motivate solving boundary placement \emph{optimally}
with \emph{formal correctness guarantees}, rather than heuristically.

\paragraph{This Work}

\emph{We present Boundary Allocation Optimizer (\tool),
a novel compilation pipeline for intermittent computing that formulates region-boundary placement and memory
allocation as a mixed-integer linear program (MILP).}
Our key idea is to encode energy-feasibility and
variable checkpointing constraints
directly in the optimization problem.
Given a control-flow graph with estimated basic block execution frequencies and
energy cost parameters, \tool produces a boundary placement and variable
allocation that minimizes expected energy overhead, and instruments the
program accordingly.
By construction, an optimal solution minimizes overhead
among all instrumentations representable by the MILP.
We further provide a formal cost semantics and prove that any feasible
solution guarantees no execution path exceeds the energy budget.
\Cref{fig:overview} illustrates the pipeline:
the compiler profiles energy costs, chunks loops, builds a summarized
control-flow graph, formulates and solves a MILP, and instruments the
program with optimal region boundaries.

We assume that programs do not use recursion or indirect calls and that
maximum loop trip counts are annotated. These assumptions are typical in embedded
software.
In our prototype, we further assume that the program is fully inlined
and therefore a single function. \Cref{sec:function-calls} describes how the
formulation can be extended to function calls.

We implement \tool as an LLVM IR-level compiler pass together with an energy
estimation pipeline that assigns conservative energy costs to LLVM basic blocks.
We also introduce a loop-chunking strategy that transforms large loops into
bounded chunks that can be treated as a single node in the MILP, improving the
quality of the instrumentation.

Our evaluation on 13 benchmarks across three capacitor sizes shows that
\tool reduces region boundary overhead and improves execution time over
both baselines while keeping MILP solve times within seconds.
For evaluation, we reimplement \rockclimb and \schematic as LLVM-based baselines for direct
comparison.
At 10\,\textmu{}F, compared to \schematic with LLVM \lstinline{O3}
optimizations, \tool is 10\% faster and encounters 52\% fewer region boundaries on average.
Under intermittent power on recorded energy-harvesting traces, \tool is
1.73$\times$--17.81$\times$ faster than the baselines.

\paragraph{Key Contributions}

The contributions of this paper are:

\begin{enumerate}[label=(\arabic*), wide=0pt, leftmargin=*]
\item We formulate region-boundary placement and memory
allocation as a mixed-integer linear program that
encodes energy-feasibility and memory-consistency constraints and
optimizes for expected energy consumption
(\cref{sec:overview,sec:algorithm}).

\item We introduce a formal cost semantics modeling energy consumption,
and use it prove that any feasible solution of the MILP guarantees energy
safety (\cref{sec:correctness-optimality}).

\item We develop a loop-chunking strategy that soundly summarizes
loop energy costs for integration into the MILP formulation
(\cref{sec:handling-loops}).

\item We implement an open-source LLVM-based compilation pipeline that
formulates and solves the MILP and instruments LLVM IR, along with an
energy profiling pipeline for deriving basic block-level energy estimates
from hardware measurements (\cref{sec:parameter-estimation}).

\item We reimplement \rockclimb and \schematic as LLVM
passes with bug fixes and evaluate all three systems on thirteen
benchmarks across three capacitor sizes and on recorded
energy-harvesting traces, showing that \tool reduces
region boundary overhead and execution time compared to both baselines
(\cref{sec:evaluation}).
\end{enumerate}

\section{Background and Motivation}
\label{sec:background}

We first review existing checkpointing strategies,
then we discuss the shortcomings of heuristic-based energy-aware region identification strategies to motivate \tool.

\subsection{Checkpointing Approaches}

\paragraph{Systems without forward progress guarantees}
Early systems such as DINO~\cite{lucia2015simpler} and
Ratchet~\cite{van2016intermittent} insert checkpoints at compile time
so that each
program segment between consecutive checkpoints is idempotent, preserving
memory consistency upon re-execution after a power failure.
Similarly, task-based programming models~\cite{alpaca, chain} require the
programmer to manually define tasks as atomic units of execution; the system
checkpoints only the task identity at each task boundary, and selectively
protects variables involved in write-after-read dependencies rather than saving
the full volatile state.
In both cases, when power fails, the system resumes from the most recent
checkpoint.

These approaches, however, guarantee neither forward progress nor efficient
execution under tight energy budgets.
Although techniques exist to reduce checkpoint frequency or
size~\cite{tiling,wario,lact}, they either delegate checkpoint placement
to the programmer~\cite{lucia2015simpler, alpaca, chain}, or do not attempt to
guarantee forward progress at all~\cite{van2016intermittent}.
Finding good placements manually is difficult because the search space is large
and optimal placement depends on hardware-specific energy characteristics that
vary across boards.
Tools exist to assist this process, but each has
limitations: CleanCut~\cite{cleancut} bisects oversized tasks without optimizing
boundary placement, potentially generating more tasks than necessary;
auto-tuning approaches~\cite{auto-tune} take a long time to find solutions,
as reported by Choi et al.~\cite{rockclimb}; and
energy debuggers~\cite{colin2017energy, fused} can only analyze paths visited
at run-time, making it difficult to guarantee forward progress on all possible
paths.

\paragraph{Forward progress via run-time checks}
To better support forward progress,
JIT-checkpointing approaches~\cite{hibernus, hibernuspp, nvp-design, coati}
use a hardware interrupt to reactively save volatile state when the voltage
crosses a predefined threshold.
These approaches require dedicated hardware support for voltage monitoring
interrupts and must reserve sufficient
energy for the checkpoint,
reducing the usable energy~\cite{ahmed2019efficient}.
Mementos~\cite{mementos}, HarvOS~\cite{harvos}, and
Chinchilla~\cite{maeng2018adaptive} remove the need for hardware
interrupt
support by combining compile-time instrumentation with run-time voltage checks,
at the cost of run-time overhead for evaluating checkpoint conditions.

\subsection{Energy-aware Region Formation}
\label{sec:rollback-free}

Recent works, \rockclimb~\cite{rockclimb} and \schematic~\cite{schematic2024},
 split a program at compile time into
\emph{regions}, each sized to fit within the device's single-charge energy
budget, and save and restore volatile program state at region boundaries.
Because the device fully replenishes its energy buffer before each region
begins, every region is guaranteed to complete within a single charge cycle,
eliminating re-execution.
These guarantees rest on a compile-time analysis that conservatively bounds
the energy cost of each region using worst-case energy consumption
techniques such as instruction-level~\cite{tiwari1994power, whole-system-worst-case-energy-consumption,
data-dependent-energy-modeling, fused}
or LLVM IR-level analysis~\cite{llvm-static-analysis}:
as long as every region's bound stays within the budget, forward progress
is ensured.

\rockclimb and \schematic differ in how they
decide where to place region boundaries.
\rockclimb takes a greedy approach: it seeds mandatory boundaries at function
call sites and loop headers, then traverses the CFG in topological order,
accumulating the per-basic-block cost from the most recent boundary and
inserting a new boundary whenever the accumulated cost would exceed the
energy budget before the next seeded boundary is reached.
\schematic replaces this greedy traversal with a per-path shortest-path
optimization: it analyzes execution paths from program entry to exit in
descending order of expected frequency and for each path builds a graph
whose nodes are the potential region boundaries along the path.
An edge connects two boundaries $c_1, c_2$ when the code segment between
them fits in the energy budget under some memory allocation, weighted by
that segment's total energy cost (restore at $c_1$, execution under the
chosen allocation, and save at $c_2$), so that the shortest path from entry
to exit yields the minimum-energy subset of boundaries together with the
matching variable allocation.
Once a path has been analyzed, its boundary and allocation decisions are
finalized; subsequent paths must respect them.

\rockclimb and \schematic achieve lower overhead compared to Chinchilla,
proving the usefulness of energy-aware region-formation algorithms.
As far as we know, \schematic is also the first approach to co-optimize memory
allocation with boundary placement: because the energy cost of a region depends
not only on which boundaries are active but also on whether each variable
resides in VM or NVM, \schematic selects the allocation for each path segment as
part of its shortest-path computation.
The only other system that automates the memory allocation decision is
Alfred~\cite{alfred}, which optimizes allocation to reduce energy consumption
but treats the checkpointing strategy as fixed input, and therefore cannot
exploit this coupling.

\subsection{Limitations of Heuristics-based Approaches}
\label{sec:limitations}

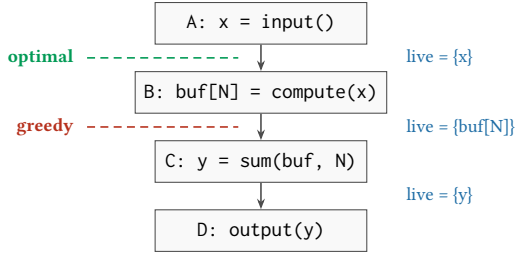
\begin{figure}[t]
\centering
\begin{tikzpicture}[
  node distance=0.6cm,
  block/.style={
    rectangle, draw=black!70, fill=gray!5,
    minimum width=2.8cm, minimum height=0.55cm,
    font=\footnotesize\ttfamily, align=left
  },
  arr/.style={-{Stealth[length=4pt,width=3pt]}, semithick, black!70},
  liveset/.style={font=\scriptsize, text=RoyalBlue!80!black, anchor=west},
  ckptlbl/.style={font=\scriptsize\bfseries, inner sep=2pt},
  ckptline/.style={densely dashed, semithick},
]

\node[block] (B1) {A: x = input()};
\node[block, below=0.35cm of B1] (B2) {B: buf{[}N{]} = compute(x)};
\node[block, below=0.35cm of B2] (B3) {C: y = sum(buf, N)};
\node[block, below=0.35cm of B3] (B4) {D: output(y)};

\draw[arr] (B1.south) -- (B2.north);
\draw[arr] (B2.south) -- (B3.north);
\draw[arr] (B3.south) -- (B4.north);

\node[liveset] at ([xshift=1.8cm]$(B1.south)!0.5!(B2.north)$) {live = \{x\}};
\node[liveset] at ([xshift=1.8cm]$(B2.south)!0.5!(B3.north)$) {live = \{buf{[}N{]}\}};
\node[liveset] at ([xshift=1.8cm]$(B3.south)!0.5!(B4.north)$) {live = \{y\}};

\draw[ckptline, BrickRed] ([xshift=-2.3cm]$(B2.south)!0.5!(B3.north)$)
  -- ([xshift=-0.3cm]$(B2.south)!0.5!(B3.north)$);
\node[ckptlbl, BrickRed, anchor=east]
  at ([xshift=-2.4cm]$(B2.south)!0.5!(B3.north)$) {greedy};

\draw[ckptline, ForestGreen] ([xshift=-2.3cm]$(B1.south)!0.5!(B2.north)$)
  -- ([xshift=-0.3cm]$(B1.south)!0.5!(B2.north)$);
\node[ckptlbl, ForestGreen, anchor=east]
  at ([xshift=-2.4cm]$(B1.south)!0.5!(B2.north)$) {optimal};

\end{tikzpicture}
\caption{Greedy boundary placement between B and C forces the large
  array \texttt{buf[N]} into NVM; placing the boundary between A and B
  avoids this cost.}
\Description{Control flow graph illustrating how greedy checkpoint placement can be suboptimal by forcing expensive checkpoints of large live variables.}
\label{fig:cfg-greedy-suboptimal}
\end{figure}

Current energy-aware region-formation approaches rely on heuristics rather than
whole-program optimization, which can lead to suboptimal boundary
placement.
Consider the four-block program in
\cref{fig:cfg-greedy-suboptimal} as an example.
It reads a scalar input \lstinline{x}, populates an array
\lstinline{buf[N]} via \lstinline{compute(x)}, sums the array into
\lstinline{y}, and outputs \lstinline{y}.
Suppose the energy budget is sufficient to execute any two adjacent
blocks but not three.

\paragraph{Greedy Placement}
\rockclimb unconditionally allocates all variables in NVM and greedily
places a boundary between B and C.
Because \lstinline{buf[N]} resides in NVM, blocks B and C must read and
write NVM for every access to the array, which is slower than accessing VM.
Allocating \lstinline{buf[N]} in VM would make these accesses cheaper,
but under this boundary placement \lstinline{buf[N]} is live at the end
of B and must be saved to and restored from NVM at the boundary, which
can be more expensive than the NVM access overhead it eliminates.

\paragraph{Placement-Allocation Coupling}
A different boundary placement avoids this tradeoff.
If the boundary is placed between A and B and \lstinline{buf[N]} is
allocated in VM, only the scalar \lstinline{x} is live at the boundary.
Thus, the checkpoint saves and restores one variable and both
\lstinline{compute(x)} and \lstinline{sum(buf, N)} access the array at
VM speed.
Reaching this solution requires jonitly considering boundary placement and
memory allocation: fixing either decision first forces a
suboptimal choice for the other.
\rockclimb cannot discover it because it fixes all variables to NVM.
\schematic can find this solution in cases like
\cref{fig:cfg-greedy-suboptimal}, as it co-optimizes allocation
with placement.

However, \schematic's heuristic does not guarantee good
solutions in the presence of complex control flow.
For example, its path-by-path commitment causes fragmentation across
shared CFG segments.
If analyzing a higher-priority path inserts a boundary on a CFG segment
that is later reused by a lower-priority path, the later path must
inherit that split even when its own best decomposition would keep the
segment within a single region.
Similarly, because each segment's allocation is committed before the
paths that share it are analyzed, a later path that traverses two
segments with different allocations can only reconcile them by
inserting a boundary. Changing either allocation is never revisited,
even when doing so would avoid the boundary.

Its non-revisiting structure also creates difficulties for
loop analysis.
To handle loops, \schematic uses conditional checkpoints: checkpoints
inserted at the loop latch that fire every $K$ iterations rather than
on every iteration.
In nested loops, an inner loop may first be analyzed under the
assumption that it requires conditional checkpointing, reserving energy
for that possibility in its cost summary.
As the \schematic paper specifies, per-block energy estimates only grow more
conservative, so the enclosing loop's analysis never revises this reserved energy.
The outer loop then overestimates the cost of the inner loop, multiplied by its
iteration, and may introduce region boundaries more than necessary.

\section{Overview}
\label{sec:overview}

We present \tool, a compilation pipeline that jointly optimizes
region-boundary placement and memory allocation via MILP.

\subsection{Assumptions}

\paragraph{Target execution model}
\tool's execution model assumes that each program is split into \emph{regions}.
At each region boundary, the system saves its modified volatile state to
non-volatile memory and waits for the device to fully recharge its capacitor.
When the capacitor is fully charged, the system restores the relevant state and
starts executing the next region.
Power failures are expected in region boundaries (mid waiting), but not at any
other points.

\paragraph{Target programs}
We assume that there are no recursive or indirect function
calls, and every loop is annotated with a known upper bound on the
number of iterations.
These restrictions are standard in WCET analysis and embedded
safety-critical coding guidelines~\cite{power10, wcec-survey},
as they ensure statically predictable execution.
As we describe in \cref{sec:handling-loops},
the loop annotations are used only to allow more aggressive optimizations,
and \tool works correctly even without them.

Our prototype inlines the target program
and instruments only the \lstinline{main} function.
Since there are no recursive or indirect calls, every user function can be
inlined, which is the case for all our benchmarks.
However, a few low-level runtime routines, such as the software-emulated
floating-point arithmetic are not inlined by LLVM.
\Cref{sec:function-calls} describes how \tool handles these calls.
That section also shows how the formulation can be extended to support more
general function calls, which we leave as future work.

\subsection{Mixed-Integer Linear Programming (MILP)}

\tool's goal is to determine where to place region boundaries and
how to allocate each variable (VM or NVM) so that every region fits
within the energy budget while minimizing expected energy consumption.
To achieve this, we use MILP, which
is an optimization problem that minimizes a linear objective
subject to linear constraints, where variables may be either integers
or real numbers.
It is a natural fit for our problem: boundary placement
and memory allocation are binary decisions, and the energy budget and overhead
are real numbers linear in those decisions.
\tool uses Gurobi~\cite{gurobi}.

\subsection{Compilation Pipeline}

\Cref{fig:overview} shows \tool's full compilation pipeline, built
around this MILP core.
The compiler first profiles each basic block's energy consumption on the target
board (Step~1).
Using the energy profile, it preprocesses loops so that the MILP formulation can
handle them effectively (Step~2; see \cref{sec:algorithm} for details).
Block execution frequencies are estimated from the preprocessed code
via profiling.
From the preprocessed program, the compiler builds a summarized control flow
graph (Step~4) and encodes the boundary placement and memory allocation problem
as a MILP (Step~5).
An off-the-shelf solver finds the optimal solution (Step~6), which specifies
where to place region boundaries and whether each variable should reside in VM
or NVM.
Finally, the compiler instruments the code with save and restore operations at
the determined boundaries (Step~7), producing the final instrumented code.

\subsection{Save and Restore Model}

While \tool
assumes that each region is guaranteed to complete within a single charge cycle,
a power failure may occur between any two regions.
the system must
save relevant variables before a region ends and properly
restore them before a new region begins.
Specifically, at each region boundary, we
\begin{enumerate*}[label=(\arabic*)]
    \item save all dirty VM variables to NVM before ending the current
    region, and
    \item restore all live-in variables from NVM to VM before starting the next
    region.
\end{enumerate*}
This model saves only variables that have been modified (dirty) within a region,
rather than unconditionally saving all registers as in
Schematic~\cite{schematic2024}.
Similarly, our restore model loads only variables that are live at the region
entry, rather than restoring all physical registers as in both \rockclimb and
Schematic.

Unlike \rockclimb's distributed checkpointing, where save logic is inserted
immediately after each variable's definition, we consolidate all saves at the
region boundary.
Distributed checkpointing is viable for \rockclimb because its regions do not
span loop iterations: loops are split such that each iteration forms its own
region, so each definition site within a region is reached at most once.
In contrast, our approach preserves loops within regions when they fit within the
energy budget, so a definition inside a loop body may be reached many times per
region.
Distributing saves in this setting would execute the save on every iteration,
making the cost proportional to the iteration count.

\subsection{Mid-Region Power Failures}

If \tool underestimates the energy consumption of a region,
a power failure can occur in the middle of the region,
a limitation shared with prior works~\cite{rockclimb, schematic2024}.
We could apply mitigations from \schematic~\cite{schematic2024}.
First, we can add a safety margin by shrinking the energy budget of each
region.
Second, the system can detect a mid-region power failure by checking whether
it restarted from the same boundary twice, and restore the program to
its initial state; if such failures persist, region placement can be
recomputed with a smaller capacitor size and deployed through an
over-the-air update.
We implement this detection with negligible overhead by recording the
last region boundary in NVM during normal execution, and comparing it with
the previously recorded boundary only on restart.
We do not implement the other mitigations, as no mid-region power failure
occurred in our experiments under real intermittent power
(\cref{sec:mid-region-failures}).

\subsection{Key Challenges}
Formulating boundary placement and memory allocation as an optimization
problem that minimizes expected energy consumption introduces several
challenges.
First, a naive MILP formulation for loops would place a region boundary
at every iteration, producing inefficient code.
We address this with \emph{loop chunking}, a preprocessing step that splits
loops into bounded inner loops, and \emph{CFG summarization}, which collapses
these inner loops into single nodes, abstracting away small cycles from the
formulation (\cref{sec:handling-loops}).
Second, tracking which variables must be saved at a region boundary is
mutually dependent with boundary placement: placing a boundary resets
dirtiness by saving
variables, while having more dirty
variables increases the cost of placing a boundary.
We resolve this circularity by encoding dirtiness propagation directly as
MILP constraints, so that the solver reasons about boundary placement and
save decisions simultaneously (\cref{sec:algorithm}).

\section{Algorithm}
\label{sec:algorithm}

This section describes how \tool formulates region-boundary placement and
memory allocation as a MILP.
After setting up the problem (\cref{sec:problem-setup}), we present
the key ideas
(\cref{sec:key-ideas}).
We then show how we instrument the program with decisions made by the MILP
solver (\cref{sec:instrumentation}), and state correctness and optimality
results (\cref{sec:correctness-optimality}).
Finally, we describe how loops are preprocessed to
reduce overheads (\cref{sec:handling-loops}).

\subsection{Problem Setup}
\label{sec:problem-setup}

The input to \tool's algorithm is a control flow graph $G = (B, E)$,
where $B$ is the set of basic blocks and
$E \subseteq B \times B$ is the set of directed edges.
We assume that there exists a unique entry block $b_0 \in B$
and the set of exit blocks $B_{\mathrm{exit}} \subseteq B$.

Each block $b$ has a set of defined variables $\defs(b)$
(either virtual register declaration or memory write)
and live-in variables
$\LiveIn(b)$, obtained from standard compiler analyses.
\tool's algorithm also takes input parameters listed at the top of
\cref{tab:milp-symbols}, estimated via profiling and measurement
(\cref{sec:parameter-estimation}).

$\variables$ denotes the set of program variables, partitioned into
\emph{candidates} ($\eligibles$) whose allocation site (VM or NVM) the solver
chooses, and \emph{fixed} variables ($\ineligibles$) that always reside in VM.
In our implementation, $\eligibles$ consists of the non-constant global
variables, and $\ineligibles$ consists of the SSA virtual-register values and
stack allocations.
SSA values
are mapped by the backend to
registers, and making them NVM candidates would forgo register allocation.
Stack allocations
do not have static addresses,
so redirecting their accesses to an NVM location requires additional work;
we leave supporting them as candidates to future work.

\begin{table}[t]
  \centering
  \caption{Constants and variables used in \tool's MILP formulation.
    Block index $b \in B$; candidate variable index $v \in \eligibles$.
    The full list of auxiliary variables and constraints is in
    \cref{sec:appendix-algorithm}.}
  \label{tab:milp-symbols}
  \small
  \setlength{\tabcolsep}{4pt}
  \begin{tabular}{@{}l p{0.7\columnwidth}@{}}
    \toprule
    \textbf{Symbol} & \textbf{Meaning} \\
    \midrule
    \multicolumn{2}{@{}l}{\emph{Inputs}} \\
    $\ebuffer \in \mathbb{R}_{\ge 0}$      & Per-region energy budget \\
    $F_b \in \mathbb{R}_{\ge 0}$           & Execution frequency of block $b$ \\
    $\Ebase_b \in \mathbb{R}_{\ge 0}$      & Energy to execute $b$ assuming all candidate variables are in VM \\
    $\deltaE_{b,v} \in \mathbb{R}_{\ge 0}$ & Extra energy when $v$ is accessed from NVM rather than VM in $b$. \\
    $\Erst_v \in \mathbb{R}_{\ge 0}$       & Cost to restore $v$ from NVM to VM\\
    $\Esv_v \in \mathbb{R}_{\ge 0}$        & Cost to save $v$ from VM to NVM \\
    $\Epro,\ \Eepi \in \mathbb{R}_{\ge 0}$ & Overhead of region start / end \\
    \midrule
    \multicolumn{2}{@{}l}{\emph{Decision variables}} \\
    $r_b \in \{0,1\}$            & A region boundary is placed at the start of $b$ \\
    $m_{b,v} \in \{0,1\}$        & $v$ is allocated in VM at $b$\\
    $s_{b,v} \in \{0,1\}$        & $v$ must be saved at the region boundary at the start of $b$ \\
    $\hat{r}_{b,v} \in \{0,1\}$  & $v$ must be restored at the region start at $b$ \\
    $d_{b,v} \in \{0,1\}$        & $v$ is written since the last region boundary at $b$ \\
    $\eaccum_b \in [0,\ebuffer]$ & Worst-case energy consumption accumulated in the current region up to $b$ \\
    \midrule
    \multicolumn{2}{@{}l}{\emph{Derived linear expressions}} \\
    $\eblk_b$   & $\Ebase_b + \sum_{v} \deltaE_{b,v}\,(1 - m_{b,v})$ (block execution cost)\\
    $\estart_b$ & $\Epro\, r_b + \sum_{v} \Erst_v\, \hat{r}_{b,v}$ (region start cost)\\
    $\eend_b$     & $\Eepi\, r_b + \sum_{v} \Esv_v\, s_{b,v}$ (region end cost)\\
    \bottomrule
  \end{tabular}
\end{table}

\subsection{Key Ideas}
\label{sec:key-ideas}

\paragraph{Objective}
The solver minimizes the expected energy consumption given memory placement
decisions and region boundary instrumentation,
where the decisions are encoded as MILP variables:
$r_b = 1$ if a boundary is placed at the start of block $b$, and
$m_{b,v} = 1$ if $v$ is allocated in VM at block $b$
(\cref{tab:milp-symbols}).
Formally, the solver minimizes a weighted sum over blocks,
where each block $b$ contributes a total cost
$\eblk_b + \estart_b + \eend_b$, weighted by its execution frequency $F_b$:
\begin{equation*}
  \min \sum_{b\in B} F_b \cdot (\eblk_b + \estart_b + \eend_b).
\end{equation*}
The three cost terms are derived linear expressions
(\cref{tab:milp-symbols}) whose values are determined by the
MILP decision variables:
$\eblk_b$ is the block execution cost, which rises as more candidate
variables are demoted to NVM (controlled by $m_{b,v}$);
$\estart_b$ and $\eend_b$ are the region-start and region-end costs,
paid only when a boundary is placed at the beginning of $b$ (i.e., $r_b = 1$), and zero
otherwise.
$\estart_b$ pays the fixed overhead $\Epro$ plus the cost of restoring
every live VM variable at the boundary ($\hat{r}_{b,v} = 1$ when
$r_b = 1$, $m_{b,v} = 1$, and $v \in \LiveIn(b)$), and $\eend_b$ pays
the fixed overhead $\Eepi$ plus the cost of saving every dirty VM
variable ($s_{b,v} = 1$, set by the dirtiness tracking described
below).
Since the weight $F_b$ is block $b$'s execution frequency,
the weighted sum is proportional to the expected energy a
run consumes.
Consequently, the solver naturally prioritizes cost reductions in frequently
executed blocks.

We minimize total energy consumption as the optimization objective.
Total energy consumption serves as a rough proxy for the
end-to-end execution time, which is the sum of active execution time
and the cumulative recharging time across all power cycles.
Active execution time tracks energy closely because per-instruction energy on
the MSP430 largely scales with the cycles executed.
Recharging time is likewise tied to energy consumption,
since it is the consumed energy divided by the average
harvesting power over that run.
This relationship is not exactly proportional, as harvesting power fluctuates,
but since harvesting power is outside our control,
total consumed energy serves as a proxy for total recharging time.
The choice of objective is in any case flexible, and a different cost
metric can be substituted without changing the formulation.

\paragraph{Energy Budget Enforcement}
Each region must complete within the energy budget $\ebuffer$.
We enforce this through constraints on the accumulated energy
$\eaccum_b$, summarized in \cref{tab:energy-constraints}.
\ref{c:accum-init-main} initializes $\eaccum_b$ at each region start;
\ref{c:accum-prop-main} propagates it along intra-region edges;
\ref{c:energy-bound-main} and \ref{c:energy-bound-exit-main} bound the total
energy at region boundaries and exit blocks, respectively.
Note that $\estart_b$ and $\eend_b$ are zero when $r_b = 0$.
Because \ref{c:accum-prop-main} and \ref{c:energy-bound-main} hold for every
edge in the CFG, any feasible solution guarantees that no region exceeds the
budget on any execution path.

\begin{table}[t]
  \centering
  \caption{\small Energy budget constraints on accumulated energy $\eaccum_b$.}
  \label{tab:energy-constraints}
  \small
  \makeatletter
  \newcommand{\ctag}[2]{(#1)\def\@currentlabel{#1}\label{#2}}
  \makeatother
  \begin{tabular}{@{}c l l@{}}
    \toprule
    & \textbf{Description} & \textbf{Constraint} \\
    \midrule
    \ctag{C1}{c:accum-init-main}
      & Reset at region start $b$
      & $\eaccum_b = \estart_b \quad \text{if } r_b = 1$ \\
    \ctag{C2}{c:accum-prop-main}
      & \begin{tabular}[t]{@{}l@{}}Propagate through \\ edge $(u,v)$ within region\end{tabular}
      & $\eaccum_v \ge \eaccum_u + \eblk_u \quad \text{if } r_v = 0$ \\
    \ctag{C3}{c:energy-bound-main}
      & Bound at every edge $(u,v)$
      & $\eaccum_u + \eblk_u + \eend_v \le \ebuffer$ \\
    \ctag{C4}{c:energy-bound-exit-main}
      & Bound at exit block $b$
      & $\eaccum_b + \eblk_b \le \ebuffer \quad \forall b \in B_{\mathrm{exit}}$ \\
    \bottomrule
  \end{tabular}
\end{table}

\paragraph{Dirtiness and Save Tracking}
Boundary placement and save decisions are mutually dependent: placing a
boundary resets dirtiness by saving all dirty VM variables, while having more
dirty variables in VM increases the cost of placing a boundary.
The solver must therefore decide both jointly, using two additional
decision variables (\cref{tab:milp-symbols}).
The dirtiness flag $d_{b,v}$ flows forward through the CFG within a region
and resets at each boundary;
$s_{b,v}$ is set when $v$ is dirty and in VM at any predecessor of the boundary
block $b$, contributing $\Esv_v$ to $\eend_{b}$.
The invariants are enforced by MILP constraints.

\paragraph{Allocation Constraints}
We introduce additional constraints on variable allocation
for performance and implementation simplicity.
First, a variable's allocation site (VM or NVM) may change only at region boundaries.
Allowing mid-region allocation changes complicates the implementation
and incurs additional runtime overhead;
\schematic avoids this as well.
Second, a variable's allocation must be the same across all regions that merge at a join point.
Without this constraint, join points would require conditional saving logic
to flush stale VM copies when a predecessor region accesses the variable directly from NVM.
This is not a significant limitation.
If different allocations would be beneficial, the solver can place a region
boundary before the join point to separate the regions.
This differs from \schematic, where an allocation mismatch between
already-committed segments forces a boundary (\cref{sec:limitations}).

We also implement a simpler variant, \lstinline{constant-alloc}, which makes
allocation globally constant per variable (replacing $m_{b,v}$ with a single
$m_v$) to improve scalability, reducing the number of MILP variables at the cost
of per-region flexibility.
Since our solver handles the finer per-region granularity within seconds on all
benchmarks, we retain it by default.
Even the coarse variant still automatically decides each variable's allocation
at compile time, improving over systems that require users to specify it manually.

\subsection{Instrumentation}
\label{sec:instrumentation}

A feasible MILP solution is not just a proof of safety; it directly
prescribes how to rewrite the program.
Given a feasible solution $\sigma$,
\tool constructs the
instrumented program $P(\sigma)$ by directly materializing the decisions encoded
in $\sigma$.

For boundary placement, when $r_b(\sigma)=1$, \tool inserts a region boundary
at the beginning of $b$.
This boundary ends the previous region and starts the next one.
Ending a region saves every dirty VM variable selected by
$s_{b,v}(\sigma)=1$ to NVM and executes a fixed region-end runtime stub.
Starting a region restores every live-in VM variable selected by
$\hat{r}_{b,v}(\sigma)=1$ from NVM to VM.
The common part of the region start, which restores the PC and stack
pointer, is implemented at the boot address; it then transfers control to
the region-specific part for the region being resumed.

For memory placement, block $b$ accesses variable $v$ from VM when
$m_{b,v}(\sigma)=1$ and from NVM otherwise.
For each $v \in \variables$ that may be allocated in VM, we
introduce a shadow VM variable $v^{\mathrm{shadow}}$, which serves as
the VM location of $v$ while its persistent copy resides in NVM.
When the solver places $v$ in VM for block $b$ (i.e.,
$m_{b,v}(\sigma)=1$), all accesses to $v$ in $b$ are redirected to
$v^{\mathrm{shadow}}$.
Because allocation changes are allowed only at region boundaries, each
variable's placement remains fixed throughout a region.

Memory placement requires handling accesses through pointers,
for which \tool relies on LLVM's alias analysis.
When an access through a pointer cannot be resolved to a unique variable by the
alias analysis, \tool conservatively excludes every variable the access may
alias from $\eligibles$ and keeps them in NVM.
This conservatism may forgo optimization but never affects correctness; in our
benchmarks, only two variables (in activity recognition)
fell back to NVM.

\subsection{Correctness and Optimality}
\label{sec:correctness-optimality}

\paragraph{Energy Budget Enforcement}

We design a cost semantics over execution traces to formally
prove energy safety.
The cost semantics, full MILP formulation, and all proofs are in
\cref{sec:appendix-cost-semantics,sec:appendix-algorithm}.
Briefly, an execution is a trace of events $\tau$, where each event is
either an MSP430 instruction or a memory access
$\mathsf{read}(\mu, v, n)$ or $\mathsf{write}(\mu, v, n)$, with
$\mu \in \{\mathrm{VM}, \mathrm{NVM}\}$ indicating whether the access
targets VM or NVM, $v$ the program variable being accessed, and $n$ the
number of bytes of $v$ that are read or written.
We treat memory reads and writes as separate events so that the model
captures the per-access cost on which memory allocation and save/restore
decisions depend.
The cost function $\operatorname{cost}(\tau)$ sums the per-event energy
costs along $\tau$.

Lemma~\ref{lem:energy-equality-cost-semantics} states the soundness
of the MILP's energy model: for any CFG path
$\pi$ inside a region, the costs $\eblk_b(\sigma)$, $\estart_b(\sigma)$, and $\eend_b(\sigma)$
that the MILP uses to track energy sum up to $\operatorname{cost}(\pi(\sigma))$,
the cost of the event trace executed along $\pi$ under $\sigma$'s memory
placement and region-boundary instrumentation.
Every energy bound the MILP proves on this sum is thus a bound on
the real energy consumed by the instrumented program.

\begin{lemma}[store=energyEqualityCostSemantics]
\label{lem:energy-equality-cost-semantics}
Let $\sigma$ be a feasible MILP solution.
For any CFG path $\pi = b_0 \to b_1 \to \dots \to b_k$ such that
$r_{b_0}(\sigma)=1$ and $r_{b_i}(\sigma)=0$ for all $1 \le i \le k$,
we have:

\(
\operatorname{cost}(\pi(\sigma))
=
\estart_{b_0}(\sigma)
+
\sum_{i=0}^{k} \eblk_{b_i}(\sigma).
\)

If $\pi' = b_0 \to b_1 \to \cdots \to b_k \to b_{k+1}$ and
$r_{b_{k+1}}(\sigma)=1$, then we additionally have:

\(
\operatorname{cost}(\pi'(\sigma))
=
\estart_{b_0}(\sigma)
+
\sum_{i=0}^{k} \eblk_{b_i}(\sigma)
+
\eend_{b_{k+1}}(\sigma).
\)
\end{lemma}

With Lemma~\ref{lem:energy-equality-cost-semantics} in hand,
Theorem~\ref{thm:region-energy-bound}
shows that the cost of any in-region path
never exceeds $\ebuffer$.
This follows from the energy accumulation constraints
(\ref{c:accum-init-main}--\ref{c:energy-bound-exit-main}):
any CFG path within a region accumulates energy monotonically via
\ref{c:accum-prop-main}, and \ref{c:energy-bound-main} and
\ref{c:energy-bound-exit-main} ensure that the total (including region-end
costs) never exceeds $\ebuffer$.

\begin{theorem}[store=regionEnergyBound]
\label{thm:region-energy-bound}
  Let $\sigma$ be a feasible MILP solution
  and $\pi$ be a CFG path of $P(\sigma)$ that starts at a region boundary.
  If $\pi$ reaches the next region boundary, or if $\pi$ reaches program exit
  without crossing another region boundary, then
  $\operatorname{cost}(\pi(\sigma)) \le \ebuffer$.
\end{theorem}

\paragraph{Optimality}
Since the solver returns a solution $\sigma^\ast$ that is guaranteed to be
optimal, the instrumented program $P(\sigma^\ast)$ minimizes expected energy overhead among
all instrumentations representable by our formulation.
This class consists of all instrumentations that place region boundaries at the
beginning of basic blocks, save live dirty VM variables at region end, restore
live VM variables at region start, and satisfy the energy budget on every CFG
path.

\subsection{Handling Loops}
\label{sec:handling-loops}

The MILP formulation handles loops without special treatment.
A cycle in the CFG forces the accumulated energy $\eaccum_b$ to grow
along~\ref{c:accum-prop-main}, eventually
violating~\ref{c:energy-bound-main}.
The solver resolves this by placing at least one region boundary inside the
loop.

This means every loop iteration contains a region boundary.
For loops whose bodies cost much less than $\ebuffer$, checkpoint overhead
dominates execution time.
We mitigate this with two techniques: \emph{CFG summarization} collapses
loops that fit within the budget into a single summary node, and
\emph{loop chunking} transforms loops that do not fit so that they become
eligible for summarization.

\paragraph{CFG Summarization}
We avoid inserting a region boundary inside a loop if the entire loop
can be executed within the energy budget.
Such a loop is collapsed into a \emph{summary node} that is visible to
the MILP as a single node, with energy cost equal to the per-iteration
cost times the trip count.
The maximum trip count is
provided as an annotation.
Block frequencies are adjusted so that the summary node carries the
original loop frequency, and the back-edge is removed, reducing the
number of cycles the MILP must reason about.

\paragraph{Loop Chunking}
Loop chunking splits loops that
may exceed
the energy budget
 into an outer loop and a bounded inner
loop that fits within the budget, making it eligible for summarization.
Given a loop with worst-case per-iteration energy $E_{\mathrm{iter}}$,
we compute the maximum number of iterations $K$ that fit in one energy
cycle, considering the overhead introduced by chunking, before the MILP
is formulated.
If $K$ covers the entire loop, no transformation is needed.
When $K \ge 2$, we split the loop into an outer loop and a $K$-bounded
inner loop, which is handled by CFG summarization.

\tool supports two chunking strategies:
when the loop has a canonical induction variable, a single exit, and a
constant trip count $N$, we rewrite the loop bounds directly
(\cref{fig:exit-rewrite}).
For loops without a canonical induction variable or with unknown trip
counts (e.g., \lstinline{while} loops), we insert a
modulo-$K$ counter that forces the loop to exit to the outer loop every
$K$ iterations, without modifying
the original exit condition
(\cref{fig:chunk-counter}).
This incurs an additional counter increment and comparison per iteration,
but applies to all loops in simplified form (obtained by the \lstinline{loop-simplify} pass in LLVM).
In both cases, the inner loop runs at most $K$ iterations
and the outer loop iterates $\lceil N/K \rceil$ times, where $N$ is the trip
count.

\begin{figure}[t]
  \centering
  \tikzset{
    loopchunkbox/.style={
      rectangle,
      rounded corners=2pt,
      draw=black!65,
      fill=gray!4,
      line width=0.4pt,
      inner xsep=4pt,
      inner ysep=3pt,
      font=\scriptsize\ttfamily,
      align=left
    },
    loopchunkarr/.style={
      -{Stealth[length=5pt,width=4pt]},
      line width=0.6pt,
      draw=RoyalBlue!75!black,
      shorten <=1pt,
      shorten >=1pt
    },
    loopsumbox/.style={
      rounded corners=1.5pt,
      draw=RoyalBlue!75!black,
      dashed,
      line width=0.5pt,
      inner xsep=3pt,
      inner ysep=2pt
    }
  }
  \begin{subfigure}[t]{0.48\columnwidth}
    \centering
    \begin{tikzpicture}
      \node[loopchunkbox] (beforeA) {%
        \begin{tabular}{@{}l@{}}
          {\bfseries for} (i=0; i<N; i++)\\
          \hspace*{1em}body(i);
        \end{tabular}%
      };
      \node[loopchunkbox, below=12pt of beforeA] (afterA) {%
        \begin{tabular}{@{}l@{}}
          {\bfseries for} (o=0; o<N; o+=K) \{\\[-1pt]
          \begin{tikzpicture}[baseline=(current bounding box.center)]
            \node[loopsumbox] (sumA) {%
              \begin{tabular}{@{}l@{}}
                \hspace*{1em}lim = min(o+K, N);\\
                \hspace*{1em}{\bfseries for} (i=o; i<lim; i++)\\
                \hspace*{2em}body(i);
              \end{tabular}%
            };
          \end{tikzpicture}\\[-1pt]
          \begin{tabular}{@{}l@{\hspace{0.5em}}l@{}}
            \}; &
            {\scriptsize\rmfamily\bfseries\color{RoyalBlue!75!black}summarized}
          \end{tabular}
        \end{tabular}%
      };
      \draw[loopchunkarr] (beforeA.south) -- (afterA.north);
    \end{tikzpicture}
    \caption{constant trip count}
    \label{fig:exit-rewrite}
  \end{subfigure}
  \hfill
  \begin{subfigure}[t]{0.48\columnwidth}
    \centering
    \begin{tikzpicture}
      \node[loopchunkbox] (beforeB) {%
        \begin{tabular}{@{}l@{}}
          {\bfseries while} (cond) \{ body(); \}
        \end{tabular}%
      };
      \node[loopchunkbox, below=12pt of beforeB] (afterB) {%
        \begin{tabular}{@{}l@{}}
          {\bfseries while} (true) \{\\[-1pt]
          \begin{tikzpicture}[baseline=(current bounding box.center)]
            \node[loopsumbox] (sumB) {%
              \begin{tabular}{@{}l@{}}
                \hspace*{1em}{\bfseries for} (c=0; c<K; c++) \{\\
                \hspace*{2em}{\bfseries if} (!cond) {\bfseries goto} done;\\
                \hspace*{2em}body();\\
                \hspace*{1em}\}
              \end{tabular}%
            };
          \end{tikzpicture}\\[-1pt]
          \begin{tabular}{@{}l@{\hspace{0.5em}}l@{}}
            \} &
            {\scriptsize\rmfamily\bfseries\color{RoyalBlue!75!black}summarized}
          \end{tabular}\\
          done:
        \end{tabular}%
      };
      \draw[loopchunkarr] (beforeB.south) -- (afterB.north);
    \end{tikzpicture}
    \caption{unknown trip count}
    \label{fig:chunk-counter}
  \end{subfigure}

  \caption{Loop chunking transforms a loop into an outer loop and a
    $K$-bounded inner loop.}
  \Description{Diagrams showing the loop chunking transformation with trip count and chunk counter strategies.}
  \label{fig:loop-chunking}
\end{figure}
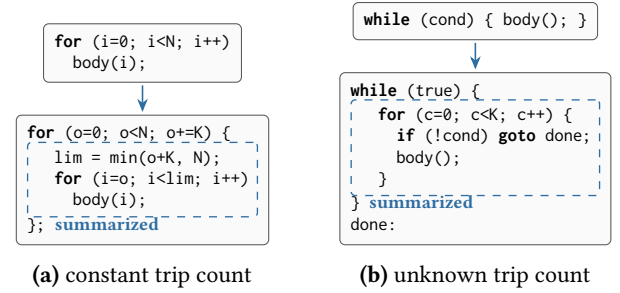

The pass processes nested loops from outermost to innermost.
If an outer loop is successfully chunked, its inner loops
are skipped, as each chunk already fits within the energy budget.
Otherwise, the pass recurses into the sub-loops.
When computing $E_{\mathrm{iter}}$ for a loop containing sub-loops,
each sub-loop's energy is conservatively estimated as its worst-case
per-iteration energy multiplied by its trip count.

While this incurs more overhead than loop unrolling, it is more general as it
supports arbitrary loop structures with unknown trip counts.
Also, loop unrolling is limited by program size bloat.
Note that our chunking strategy is semantically equivalent to the conditional
checkpointing in \schematic, but we avoid it as it would require
ad-hoc complexities in the MILP formulation.
Loop chunking and CFG summarization keep the MILP formulation unchanged, taking
advantage of its flexibility and safety guarantees for free.

\subsection{Handling Function Calls}
\label{sec:function-calls}

Our formulation so far assumes that all functions are inlined into a single
procedure to enable globally optimal placement.
Handling function calls is still important, however, as some functions are not
inlined, and per-function compilation is necessary to scale to larger programs.

A function whose worst-case energy fits within the budget needs no
boundary inside it; we call it \emph{boundary-free}.
\tool neither runs MILP on nor instruments
such a function,
and instead treats each call
to it
as a single instruction
whose cost is the function's worst-case energy.
This also avoids unnecessary cost for callers,
as placing boundaries inside such a function would
otherwise force every call site to pay for the cost,
even when the caller does not need it.

In our evaluation, most calls that remain after inlining are boundary-free.
Notable exceptions are \lstinline{memcpy} and \lstinline{memset}, which we
currently treat as boundary-free by assuming the transfer size is small enough
to fit within the energy budget (\cref{sec:parameter-estimation}).
Ideally, they would be handled as boundary-capable functions, as described
below.

A function whose worst-case energy does not fit within
the budget is \emph{boundary-capable}.
Such a function
needs region boundaries inside it,
so \tool must solve its MILP.
The result must then be composed with the MILP solutions of its callers.
We outline this extension here and leave its implementation to future work.
The extension would solve each such function's MILP and produce three summary
constants: the worst-case energy from entry to the first region boundary,
from the last region boundary to return,
and from entry to return along paths that cross no boundary.
The caller's MILP would then use these constants to bound the energy of regions
that include calls to already-solved functions, without reasoning about the
callees' internal structure.
Since there are no recursive or indirect calls (\cref{sec:overview}), the call
graph is a DAG, so each callee could be processed before its callers.

Beyond energy, a power failure inside a callee may erase the caller's volatile
state, so \tool must ensure it is saved and restored.
First, dirty shadow variables have to be committed to NVM before a
boundary-capable call,
so each function's MILP can assume all candidate variables are current in NVM at entry.
Second, the caller's return address has to be saved to NVM before each
boundary-capable call and restored when the callee returns, since a
power failure inside the callee erases the stack in SRAM.
Both costs are reflected in the caller's MILP.
Spilled variables that are live across the call need no special mechanism,
as the caller's MILP already accounts for saving and restoring them.

Because every callee summary is a conservative bound, this extension
preserves safety.
While global optimality is not guaranteed, each function is
solved optimally in isolation.

\section{Parameter Estimation}
\label{sec:parameter-estimation}

Our algorithm requires energy estimates for each program segment in order to
determine region-boundary placement.
We obtain these estimates in two stages:
first, we construct an instruction-level energy model for
MSP430 assembly (\cref{sec:energy-profiling-assembly});
then, we use this model to estimate the energy consumption of each LLVM IR basic
block (\cref{sec:energy-profiling-ir}).
We also estimate basic block execution frequencies
for the MILP objective (\cref{sec:frequency-estimation}).

\subsection{Estimating Energy Consumption of Assembly}
\label{sec:energy-profiling-assembly}

\paragraph{Energy Model}
We construct an instruction-level energy model for MSP430 assembly programs
following the approach of Tiwari et al.~\cite{tiwari1994power},
which assigns a fixed cost to each instruction type.
Each instruction is characterized by its
$(\mathit{opcode}, \mathit{addrmode}_1, \dots, \mathit{addrmode}_n)$ tuple.
The cost of a basic block is then the sum of
the costs of its instructions.

As noted in \cref{sec:function-calls}, we treat
\lstinline{memcpy} and \lstinline{memset} as boundary-free in our benchmarks.
This simplifies our current implementation
and avoids unnecessary region boundaries inside these functions.
Their energy consumption grows with the \lstinline{size} argument
and can exceed the energy budget in general,
but every call in our benchmarks has a small bounded size,
and most of them are compile-time constants.
We model each call with a linear cost function
$a \cdot \mathit{size} + b$,
whose parameters $a$ and $b$ are estimated jointly
with the per-instruction energy costs.
When the \lstinline{size} argument is not known at compile time,
we use a user-provided conservative upper bound (64 for our benchmarks).
In either case, the cost stays below the energy budget.
Supporting unbounded sizes would require treating these calls as
boundary-capable (\cref{sec:function-calls}).

\paragraph{Estimation Method}

We estimate the per-instruction energy costs by solving a linear program
that finds the smallest costs such that predicted energy
always exceeds measured energy.
Each program for training executes a known sequence of instructions, and we
measure its total energy consumption using the profiling setup described below.
To obtain the exact instruction counts for each program, we implement an
assembly interpreter, which replays each benchmark and records the
number of executions of each instruction type.
Given $m$ programs, let $\mathbf{a}_j \in \mathbb{N}^k$ be the vector of
instruction counts for program $j$ and $E_j$ the measured energy.  Note that
$\mathbf{a}_j$ includes not only the target instruction but also the jump,
counter increment, and comparison instructions that implement the enclosing
loop.
We solve
\begin{equation*}
  \min_{\mathbf{e} \ge 0,\, \mathbf{s} \ge 0} \mbox{$\sum_{j=1}^{m} s_j$} \quad \text{subject to} \quad A\mathbf{e} = \mathbf{E} + \mathbf{s},
\end{equation*}
where $A = [\mathbf{a}_1, \dots, \mathbf{a}_m]^\top$
and $\mathbf{e} \in \mathbb{R}^k$ is the vector of per-instruction energy costs,
and $s_j$ is the overestimation slack for benchmark $j$.
This formulation ensures that the predicted energy of every program
exceeds its measured energy (i.e., $\mathbf{s} \ge 0$),
and finds the tightest such upper bound
by minimizing the total overestimation $\sum_j s_j$.

\paragraph{Profiling Setup}

We power the MSP430 board with the Otii Ace Pro
at a fixed \qty{3.3}{\volt} supply voltage
and sample its current draw at 50\,ksp/s.
We use a GPI pin to measure the execution interval of each program and
integrate the resulting power trace over this interval to obtain its energy
consumption.
Each benchmark program repeats a single target instruction in a nested loop
to accumulate a large number of executions per measurement,
reducing the
noise.
Instructions that cannot be repeated naively
(e.g., \lstinline{br}, \lstinline{push})
use specialized programs that repeat them in a controlled manner.
For intrinsic functions, we measure worst-case execution
by supplying inputs that maximize the number of iterations;
for memory functions (\lstinline{memcpy}, \lstinline{memset}),
we vary the \lstinline{size} argument to fit the linear model above.

\paragraph{Validation Results}
We validate our energy model on four synthetic microbenchmarks
(\cref{sec:appendix-microbenchmarks}),
each repeating a small basic block so that its energy consumption becomes measurable.
The benchmarks cover progressively more complex addressing modes:
register-only arithmetic ($-0.46\%$ error),
indexed load with computation ($+0.07\%$),
indexed load--compute--store ($-0.76\%$),
and pointer walking with auto-increment ($-1.19\%$).
The mean absolute error across the four benchmarks is $0.62\%$
(\cref{tab:microbenchmark-error}).
The
underestimation is primarily due to the data-switching
effect~\cite{ascia2001instruction}, i.e., the additional energy consumed
when operand bit patterns change between consecutive instructions.
We leave mitigating this effect to future work,
which is orthogonal to our placement algorithm.
This error does not affect our evaluation results (\cref{sec:evaluation}),
because we use the same energy model for all systems, and we evaluate on
continuous power.

\subsection{IR-Level Energy Estimation}
\label{sec:energy-profiling-ir}

Because machine-level instrumentation may interfere with already-applied optimizations such as register allocation, \tool and most compiler-assisted tools in intermittent computing
\cite{schematic2024,cleancut,maeng2018adaptive,erata2023etap} instrument at the LLVM IR level.
Thus, we need to map the per-instruction energy costs from the previous subsection, to estimate energy at the LLVM IR
level~\cite{lattner2004llvm}.

We perform an additional compilation of the program
to assembly with debug information enabled,
following the approach of \cite{llvm-static-analysis}.
The DWARF debug information
yields a mapping from each IR
block to its corresponding group of assembly instructions.
The energy cost of each LLVM IR basic block is then the sum of the
per-instruction costs of its corresponding assembly instructions,
as determined by the model in the previous subsection.
Unlike instruction-level mappings between IR and assembly,
a block-level mapping is sufficient for our algorithm
and is more robust against backend transformations.
Some LLVM IR basic blocks have no corresponding assembly block (e.g., blocks
containing only PHI nodes or unconditional jumps); we assign them zero
energy cost, as they correspond to no-ops in practice.

A limitation of this approach is that inserting instrumentation may
perturb backend code generation.
As far as we know, existing approaches at the LLVM IR level either
use a less accurate IR-level energy model~\cite{llvm-static-analysis} or share
the same limitation.
This can be mitigated by recompiling and re-deriving costs from the final
assembly after instrumentation, and re-solving MILP if any region exceeds the
budget.
We leave this to future work.

\subsection{Frequency Estimation}
\label{sec:frequency-estimation}
Our MILP objective (\cref{sec:algorithm}) weights the cost
of each basic block by its execution frequency $F_b$,
so that the solver prioritizes reducing overhead
in frequently executed blocks.
To estimate the execution frequency of each basic block, we compile the program
with an instrumentation pass that inserts a counter for each basic block and run
the instrumented program on a host machine.
The benchmarks we evaluate are non-deterministic.
For programs with input-dependent or non-deterministic behavior, one can either
run the program multiple times with inputs sampled from a representative
distribution, or use symbolic execution techniques~\cite{erata2023etap}.


\section{Evaluation}
\label{sec:evaluation}
We evaluate \tool on a set of embedded benchmarks and compare it
against existing energy-aware region-formation approaches.
We run two sets of experiments and measure execution time: one under continuous
power and the other while replaying recorded RFID energy-harvesting traces.
We also report \tool's instrumentation overhead, its sensitivity to trip-count
annotations, and its MILP solving time.

\subsection{Evaluation Setup}
\label{sec:evaluation-setup}

\subsubsection{Systems Under Comparison}

We compare \tool against existing energy-aware region-formation approaches,
\rockclimb~\cite{rockclimb} and \schematic~\cite{schematic2024}.
We implement all three systems using the LLVM compiler infrastructure
at the \lstinline{O3} optimization level.
We also evaluate \schematic at \lstinline{O0}
to match its reference implementation.
All systems share the same energy estimation model
introduced in \cref{sec:parameter-estimation}.

\paragraph{\tool}
We implement \tool as an LLVM pass. \tool operates at basic block granularity
and does not support splitting basic blocks.
This is not a practical limitation:
\schematic explicitly assumes that every basic block fits in the energy budget,
and in practice all basic blocks are much smaller than the energy budget,
so we do not lose control compared to
instruction-level placement in \rockclimb.

\paragraph{\rockclimb}
We reimplement \rockclimb as an LLVM MIR-level (Machine Intermediate Representation) pass, as no open-source implementation is available.
We chose MIR because \rockclimb requires instruction-level cycle counting
and identification of physical registers for distributed checkpointing,
neither of which is available at the LLVM IR level.
\rockclimb also requires loop unrolling to mitigate the overhead of region boundaries
within loops. We use a maximum unrolling factor of 9,
because it was the largest unrolling factor that we could use
to compile and link the program without exceeding the code size limit.\footnote{In the default linker script, the maximum code size is \qty{48}{\kilo\byte}.
Even with a modified linker script, the hard limit is \qty{64}{\kilo\byte}
because LLVM does not support the MSP430X 20-bit addressing mode.
}

\paragraph{\schematic}
We reimplement \schematic as an LLVM IR pass,
because the reference implementation has limitations
that prevent a direct comparison:
\begin{enumerate*}[label=(\arabic*)]
    \item it produces simulation-only IR
    that cannot be directly compiled to run on hardware,
    \item it is hardcoded for \lstinline{O0} optimization
    because the algorithm requires all variables
    to reside on the stack (suppressing \lstinline{mem2reg}), and
    \item it is implemented as ad-hoc Python scripts
    that
    modify LLVM IR textually,
    rather than as a standard LLVM pass,
    making it difficult to extend beyond \lstinline{O0}.
\end{enumerate*}
To isolate the effect of \schematic's algorithm from the effect
of disabled optimizations, we also evaluate \schematic at \lstinline{O3}
(hereafter \schematicO), where its algorithm decides on the placement of
unpromoted stack variables and global variables.
We also fix soundness bugs in the reference implementation
\cref{sec:schematic-bug} for a fair comparison.

\subsubsection{Platform and Configuration}

Our target platform is the MSP430FR5994, which has \qty{256}{\kilo\byte} of FRAM
(NVM) and \qty{4}{\kilo\byte} of SRAM (VM), running at \qty{16}{\mega\hertz}.
We evaluate using energy budgets computed from three capacitor sizes:
\qty{5}{\micro\farad}, \qty{10}{\micro\farad}, and \qty{50}{\micro\farad}.
The energy budget for each configuration is computed as $E = \frac{1}{2} C
(V_{\max}^2 - V_{\min}^2)$, where $V_{\max} = \qty{3.6}{\volt}$ and $V_{\min} = \qty{1.8}{\volt}$
represent the maximum and the minimum working voltage of the board, respectively.

\subsubsection{Benchmarks}
We evaluate thirteen benchmarks: activity recognition, AES,
bitcount, ChaCha20, CRC, cuckoo filter, Dijkstra, Poly1305, quicksort, RSA,
sensor fusion, SHA-256, and stringsearch.\footnote{We modified
stringsearch to store its patterns in a character array rather than an array
of pointers, because our current profiling toolchain runs the program on the
host, whose pointer size differs from the MSP430's. This is a limitation of
our prototype tooling, not of the algorithm.}
They span cryptography, sensing and classification, and general-purpose
data-processing kernels.
All are drawn from UlSWaP-Bench~\cite{ulswap-bench} except cuckoo filter, a
standard intermittent-computing benchmark~\cite{alpaca,maeng2018adaptive}.
Most omitted UlSWaP-Bench programs exceed the code/data address range
supported by LLVM's MSP430 backend.
We also omit fft and basicmath because their math-library calls
(\lstinline{sin}/\lstinline{cos}) exceed the energy budgets and require
boundary-capable function support (\cref{sec:function-calls}).
We configure input sizes and repetition counts so each benchmark spans multiple
power cycles yet completes within a few seconds, keeping evaluation time
manageable without changing program structure.

We annotated loops with trip-count bounds
and inlined all functions to expose the full program to each system's analysis.
It is safe to do so as none of the benchmarks have recursive calls.
On eight of the thirteen benchmarks the compiler already inlines every
function at \lstinline{O3}. On the remaining five it grows code size by
4--126\%, at most \qty{4.2}{\kilo\byte}.

\begin{figure*}[t]
    \centering
    \captionsetup[subfigure]{skip=0pt}
    \begin{subfigure}{\textwidth}
        \centering
        \includegraphics[width=\linewidth]{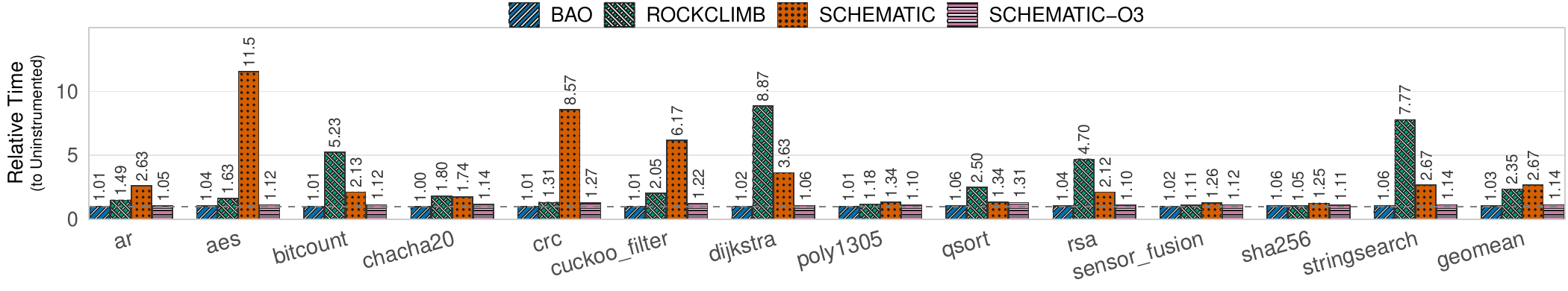}
        \caption{Normalized execution time with a 10\,\textmu{}F capacitor.}
        \Description{Bar chart comparing normalized execution time across benchmarks with a 10 microfarad capacitor.}
        \label{fig:execution-time-10uF}
    \end{subfigure}
    \begin{subfigure}{\textwidth}
        \centering
        \includegraphics[width=\linewidth]{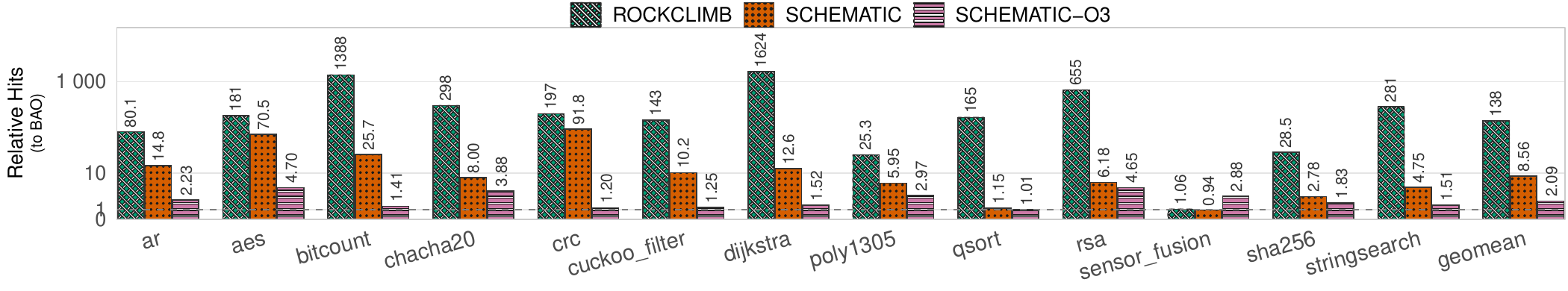}
        \caption{Normalized runtime region boundary hits with a 10\,\textmu{}F capacitor.}
        \Description{Bar chart comparing normalized runtime region boundary hits across benchmarks with a 10 microfarad capacitor.}
        \label{fig:runtime-region-boundary-calls-10uF}
    \end{subfigure}
    \vspace{-10pt}
    \captionsetup{skip=0pt}
    \caption{Normalized execution time and runtime region boundary hits with a 10\,\textmu{}F capacitor.}
\end{figure*}

\subsection{Execution Time and Region Boundary Hits}
\label{sec:continuous-power}
We measure the execution time and the number of region boundary hits at run-time
under continuous power.
The measured execution time includes saving and restoring state but excludes
power-off and recharge time.
Under intermittent power, execution time can be approximated as
$T_{\mathrm{expected}} = T_{\mathrm{base}} + R \cdot \bar{T}_{\mathrm{wait}}$,
where $T_{\mathrm{base}}$ is the measured execution time,
$R$ is the number of region boundary hits at run-time, and
$\bar{T}_{\mathrm{wait}}$ is the average additional delay per boundary hit due to
recharging and any power-off and recovery.
$\bar{T}_{\mathrm{wait}}$ depends on energy-harvesting conditions.
For a given $\bar{T}_{\mathrm{wait}}$, reducing both $T_{\mathrm{base}}$ and $R$
reduces expected execution time.

\Cref{fig:execution-time-10uF} shows execution time at
\qty{10}{\micro\farad}, normalized to uninstrumented \lstinline{O3} code.
\Cref{fig:runtime-region-boundary-calls-10uF} shows the number of region
boundary hits at \qty{10}{\micro\farad}, normalized to \tool.
Results for the other capacitor sizes are in \cref{sec:appendix-evaluation}.
On average, \tool outperforms all baselines in both execution time and
region boundary hits, and the trend is consistent across capacitor sizes:
at \qty{5}{\micro\farad}, the geometric mean (geomean) execution time of
\rockclimb, \schematic, and \schematicO
relative to \tool is
2.17$\times$, 2.72$\times$, and 1.11$\times$, respectively;
at \qty{10}{\micro\farad}, it is
2.29$\times$, 2.60$\times$, and 1.11$\times$;
at \qty{50}{\micro\farad}, it is
2.34$\times$, 2.56$\times$, and 1.07$\times$.
We discuss the \qty{10}{\micro\farad} results in detail below.

The faithful \schematic reimplementation at \lstinline{O0} is
on average 2.60$\times$ slower than \tool and incurs 8.56$\times$ more region
boundary hits.
This overhead is primarily due to the overall slowdown of the program:
without compiler optimizations, the program consumes more energy overall,
requiring more region boundaries.
While \schematic has its own memory allocation strategy, it can at best use SRAM in place of
FRAM and cannot utilize physical registers.
Furthermore, suppressing \lstinline{mem2reg} also disables all other compiler
optimizations that depend on virtual registers.
In fact, uninstrumented execution at \lstinline{O0} is
2.56$\times$ slower than uninstrumented execution at \lstinline{O3},
even when the stack is in SRAM.

Compared to \rockclimb, \tool
reduces execution time by
56\%
 and incurs 99\% fewer region boundary hits,
primarily due to the handling of loops.
As reported by its authors, \rockclimb's loop unrolling does not apply to
loops with unknown trip counts,
and the unrolling factor is limited for practical reasons, e.g., code size.
In such cases, \rockclimb inserts more region boundaries than necessary.
In contrast, \tool operates on a summarized CFG and avoids this overhead.
When we tested with a larger unrolling factor, 64,
on a small, loop-heavy benchmark, CRC,
\rockclimb's execution time was 19\% lower than with the factor of 9,
making it faster than \schematicO, though still 5\% slower than \tool.

\schematicO is the most direct baseline,
as it isolates the effect of the placement algorithm.
\tool is 10\% faster and produces 52\% fewer region boundary hits on average.
The main reason is that \schematic's energy accounting for loops is conservative.
Because the conditional checkpoint fires every $K$ iterations,
\schematic must reserve $(K{-}1) \times E_{\mathrm{iter}}$ of energy
(where $E_{\mathrm{iter}}$ is the per-iteration energy cost)
both before and after the loop to account for the worst case
in which up to $K{-}1$ iterations execute before the next checkpoint.
In nested loops, this
propagates to the enclosing loop,
forcing \schematic to place more region boundaries
than necessary.

Sensor fusion is the only benchmark where \tool does not
incur the fewest region boundary hits.
This is solely an artifact of energy estimation:
at \lstinline{O3}, the backend tail-duplicates the software floating-point calls
that dominate the hot loop's energy, and our energy profiler maps all of the
duplicated machine code back to the same IR basic block,
leading to overestimation of such a loop's energy.
This affects every system compiled at \lstinline{O3}; \schematicO incurs
3$\times$ more hits than \schematic at \lstinline{O0}.

\tool's advantage over \schematicO in region boundary hits is smallest on
quicksort.
This is due to the benchmark:
a single iteration of the outer sorting loop consumes 87\% of the energy
budget, so any placement must hit at least one region boundary per outer
iteration, and both systems are close to this lower bound.
\tool is nevertheless 19\% faster than \schematicO on quicksort,
because it saves and restores variables at virtual register granularity
rather than saving and restoring the entire set of physical registers at each
region boundary.

\subsection{Execution Time under Intermittent Power}
\label{sec:intermittent-power}

\subsubsection{Trace Replay Setup}

\paragraph{Hardware Setup}
We run the instrumented programs under intermittent power on an
MSP430FR5994, using an Otii Ace Pro that replays recorded energy-harvesting
voltage traces.
Since the Otii is a stiff voltage source,
we place a \qty{560}{\ohm}
series resistor to limit the charging current, and a Schottky diode to block
current back-feed into Otii.
We test on the board's intrinsic supply capacitance:
the decoupling capacitors on the target power rail total
\qty{11.5}{\micro\farad}~\cite{slau678}, which we also verified empirically.

\paragraph{Trace Replay}
We derive our traces from the ten RFID energy-harvesting traces distributed with
Mementos's trace-driven simulator~\cite{mementos}.
We replay each trace at $10\times$ speed to expose each run to frequent voltage
drops.
The raw recordings stay above or below the target's operating voltage for a
second or more at a time, so a real-time replay would leave most runs with only
a handful of voltage fluctuations---close to what our continuous-power
experiments already cover.
Then, we downsample each trace
from \qty{1}{\kilo\hertz} to \qty{50}{\hertz}, the highest rate at
which
the Otii replays reliably; the voltage drops survive
downsampling.
Each trace is clipped to \qty{3.6}{\volt}, the maximum operating
voltage of the board, and repeated back to back so that the replay outlasts
every benchmark.
Trace 3 never reaches the board's operating voltage range; we scale it
up before clipping until its 99th-percentile voltage reaches \qty{3.6}{\volt}.
\Cref{fig:harvesting-traces} shows the preprocessed traces.

\begin{figure}
    \centering
    \includegraphics[width=\linewidth]{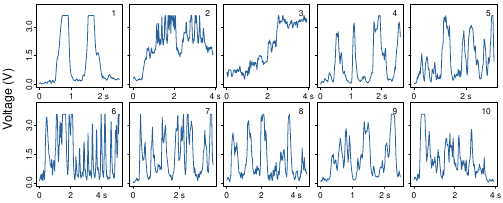}
    \caption{The energy-harvesting traces used in experiments.}
    \Description{Grid of ten voltage-versus-time plots showing the preprocessed energy-harvesting traces used for replay.}
    \label{fig:harvesting-traces}
\end{figure}

\paragraph{Measurement}
We measure execution time from
GPIO pulses emitted at benchmark start and completion,
captured with a Saleae logic analyzer.
The measured time now includes capacitor recharge time, both at region
boundaries and after power failures.
We also record how many boundary hits required waiting to recharge, how many
recoveries from power failure occurred, and whether the target ever lost power
mid-region.

\begin{figure*}[t]
    \centering
    \includegraphics[width=\textwidth]{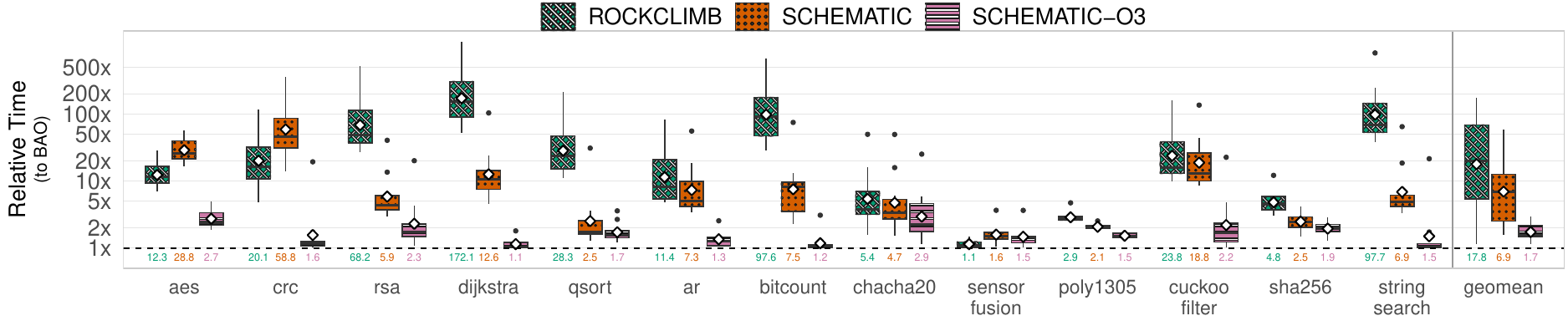}
    \caption{Execution time under intermittent power, normalized to \tool's
    run on the same trace and shown on a log scale. Each box summarizes the
    ten traces; the number below each box is the geomean.}
    \Description{Box plot comparing execution time under intermittent power, normalized to Bao, across benchmarks over ten energy-harvesting traces.}
    \label{fig:intermittent-execution-time}
    \vspace{-.25cm}
\end{figure*}

\begin{table}
  \centering
  \caption{Share (\%) of time under intermittent power
  spent on execution ($T_{\mathrm{exec}}$), checkpointing ($T_{\mathrm{ckpt}}$), and
  waiting for energy ($T_{\mathrm{wait}}$). A wait includes boot and
  recovery if power fails during it. Rows may not sum to 100 due to rounding.}
  \label{tab:intermittent-decomposition}
  \footnotesize
  \setlength{\tabcolsep}{3pt}
  \begin{tabular*}{\linewidth}{l@{\extracolsep{\fill}}rrrr}
    \toprule
                        & \tool & \rockclimb & \schematic & \schematicO \\ \midrule
    $T_{\mathrm{exec}}$ & 8.2   & 1.0        & 1.5        & 4.9 \\
    $T_{\mathrm{ckpt}}$ & 0.6   & 1.7        & 1.7        & 1.2 \\
    $T_{\mathrm{wait}}$ & 91.2  & 97.3       & 96.9       & 93.9 \\
    \bottomrule
  \end{tabular*}
  \vspace{-.25cm}
\end{table}

\subsubsection{Results}
\label{sec:mid-region-failures}
All runs completed with the correct result.
Across all runs, \tool waits at 7--152 boundaries per run (mean 23) and
experiences at most 24 power failures (mean 0.8).
Under the same traces, \rockclimb, \schematic, and \schematicO wait at 302,
125, and 38 boundaries per run, and experience 11.7, 5.1, and 1.3
power failures on average.
Power failures are rarer than waits for every system, as they
happen only when the supply voltage stays low long enough that leakage
during the wait depletes the capacitor.

\Cref{fig:intermittent-execution-time} shows
execution-time ratios relative to \tool
across the ten traces for
each benchmark, with geomeans annotated.
\tool-generated code is faster than every baseline on every run.
The geomean speedups over \rockclimb, \schematic, and \schematicO
are 17.8$\times$, 6.9$\times$, and 1.7$\times$, respectively.
\Cref{tab:intermittent-decomposition} breaks execution time down by
the measured wait time ($T_{\mathrm{wait}}$), the uninstrumented
program's time under continuous power ($T_{\mathrm{exec}}$), and the
instrumented program's extra time under continuous power
($T_{\mathrm{ckpt}}$), including each boundary's save, restore, and voltage check.
Waiting takes 91\% of \tool's time and 94--97\% of the baselines', so the
differences come largely from wait time, which grows
with the number of waits.

No run of any system lost power in the middle of a region,
indicating that the shared energy model is sufficiently conservative
for the tested harvesting conditions.
Part of this margin is by design---the energy estimation deliberately
over-approximates (\cref{sec:parameter-estimation})---but part comes
from the profiling setup.
The energy parameters are measured at a fixed \qty{3.3}{\volt} supply,
but the MCU
can operate at voltages well below \qty{3.3}{\volt},
where each instruction consumes less energy.
Developing a tighter, voltage-aware energy model
is future work.

\subsection{Instrumentation and Chunking Overhead}
\label{sec:chunking-overhead}
To quantify the overhead of \tool's program transformations,
we compare against an uninstrumented O3 baseline under continuous power,
with the 10\,\textmu{}F capacitor.
We measure two builds:
a chunk-only build, which applies loop chunking but inserts no instrumentation,
and the full \tool build, which additionally saves and restores state at
region boundaries.
Loop chunking alone increases execution time by 1.6\% (geomean)
(7.9\% maximum, on stringsearch).
Full \tool instrumentation increases execution time by 2.6\% (geomean)
(6.3\% maximum, on stringsearch), as shown in \cref{fig:execution-time-10uF}.
Since runtime overhead is inherent to any intermittent execution
system, this is an acceptable cost for \tool's correctness guarantees.

\subsection{Impact of Maximum Trip Count Annotations}
\label{sec:trip-count}
As discussed in \cref{sec:handling-loops}, \tool works without
  maximum trip count annotations for loops,
  at the cost of generating less efficient boundary placement.
To quantify this cost, we
recompile our benchmarks
without annotations, using the \SI{10}{\micro\farad} capacitor.
All 13 benchmarks compile successfully, and tests confirm the generated code
remains correct.
For 8 of the 13 benchmarks, \tool infers every loop bound statically and
therefore produces results identical to those with trip count annotations.
The remaining five benchmarks see execution time increase by
2.15$\times$ (geomean) (5.11$\times$ max, on Dijkstra)
under continuous power, with region boundary hits increasing by
51$\times$ (geomean).
The most significant source of overhead is that without inner-loop
trip count bounds, outer loops require a region boundary at every
iteration.

\begin{table}[t]
\centering
\caption{MILP compilation statistics of \tool with the \qty{10}{\micro\farad} capacitor.
Blocks and edges refer to those in the summarized CFG.
All solves reached an optimal solution.}
\label{tab:milp-compile-stats}
\begin{adjustbox}{max width=\linewidth}
\begin{tabular}{lrrrrr}
\toprule
\textbf{Benchmark} & \textbf{Blocks }& \textbf{Edges }& \textbf{Vars.} & \textbf{Constr.} & \textbf{Sol. Time (s)} \\
\midrule
ar             & \num{8} & \num{9} & \num{533} & \num{1986} & \num{0.025} \\
aes            & \num{9} & \num{10} & \num{2475} & \num{10033} & \num{0.027} \\
bitcount       & \num{36} & \num{51} & \num{5750} & \num{25057} & \num{0.086} \\
chacha20       & \num{4} & \num{4} & \num{495} & \num{1907} & \num{0.004} \\
crc            & \num{6} & \num{6} & \num{117} & \num{414} & \num{0.003} \\
cuckoo\_filter & \num{7} & \num{8} & \num{432} & \num{1710} & \num{0.005} \\
dijkstra       & \num{9} & \num{12} & \num{975} & \num{3937} & \num{0.021} \\
poly1305       & \num{4} & \num{4} & \num{559} & \num{2135} & \num{0.004} \\
qsort          & \num{12} & \num{15} & \num{324} & \num{1264} & \num{0.006} \\
rsa            & \num{26} & \num{33} & \num{4980} & \num{20658} & \num{0.105} \\
sensor\_fusion & \num{4} & \num{4} & \num{431} & \num{1696} & \num{0.004} \\
sha256         & \num{14} & \num{18} & \num{4083} & \num{16818} & \num{0.079} \\
stringsearch   & \num{18} & \num{24} & \num{955} & \num{3869} & \num{0.225} \\
\bottomrule
\end{tabular}
\end{adjustbox}
\end{table}

\subsection{MILP Solving Time}
\label{sec:milp-solving-time}

\subsubsection{Compilation Statistics}
\Cref{tab:milp-compile-stats} shows the MILP problem sizes and solving
times of \tool with the \qty{10}{\micro\farad} capacitor,
and \cref{tab:milp-compile-stats-appendix} in \cref{sec:appendix-evaluation}
reports all three capacitor sizes.
We set no time limit, and every solve terminated with a guaranteed
optimal solution.
Larger energy budgets allow CFG summarization to merge more blocks,
reducing MILP problem size and solving time.
Every benchmark solves within
\qty{0.23}{\second} (mean \qty{46}{\milli\second}) and
\qty{0.37}{\second} (mean \qty{49}{\milli\second}),
at \qty{10}{\micro\farad} and \qty{50}{\micro\farad}, respectively.
At \qty{5}{\micro\farad}, all but sensor fusion solve within
\qty{0.24}{\second} (mean \qty{59}{\milli\second}).
Sensor fusion is a hard-coded decision tree. After inlining, its MILP problem has \num{24449} variables and takes
\qty{77}{\second} to solve.
However, a capacitor this small is impractical, and with a
larger budget CFG summarization reduces the problem size
(\cref{tab:milp-compile-stats}).
This is faster than
tuning-based approaches~\cite{auto-tune, cleancut,
maeng2018adaptive}, which take hours as measured by Choi et al.~\cite{rockclimb}.

\subsubsection{Speed Ups}
\label{sec:speeding-up-milp}
We also evaluated the effect of using the \lstinline{constant-alloc} mode on \tool's MILP solve time,
using the same benchmarks.
Across the 39 benchmark--capacitor pairs, the
\lstinline{constant-alloc} mode reduced the geomean solve time by
\SI{25.0}{\percent} at \SI{5}{\micro\farad},
\SI{23.8}{\percent} at \SI{10}{\micro\farad}, and
\SI{28.8}{\percent} at \SI{50}{\micro\farad}.
The restriction changed the boundary placement or the memory allocation on 9 pairs.
Measurements show that these nine configurations
ran at most \SI{3.5}{\percent} slower than the per-region solution
(cuckoo filter at \SI{5}{\micro\farad}).
These results suggest that \lstinline{constant-alloc} offers a practical
trade-off, reducing solve time with little impact on solution quality.

\section{Conclusion and Future Work}
This paper presented \tool, a compilation pipeline for
intermittent computing that uses MILP to jointly optimize
region boundaries and memory allocation.
By encoding energy feasibility and checkpointing decisions directly in
the optimization, \tool replaces heuristics with a global formulation
whose feasible solutions provably guarantee energy safety and whose optimal
solution minimizes expected energy consumption.
Our evaluation shows that this approach reduces overheads relative to prior
systems.

\tool's clean formulation also makes it easy to extend.
For example, freshness and temporal consistency could be supported by
constraining $r_b = 0$ on every path between blocks that must execute without a
power failure, which tools such as Ocelot~\cite{ocelot} can derive from user
annotations.
Another direction is to let time-critical regions resume after recharging only
enough energy to complete the region rather than waiting for a full recharge,
which could be achieved by making $\ebuffer$ a variable bounded by the device's
energy capacity.

\section*{Acknowledgements}

This material is based upon work supported by the
\grantsponsor{NSF}{National Science Foundation}{https://doi.org/10.13039/100000001}
under Grant Nos.~\grantnum{NSF}{2525102} and~\grantnum{NSF}{2311983},
and by the
\grantsponsor{ONR}{Office of Naval Research}{https://www.onr.navy.mil/}
under Award No.~\grantnum{ONR}{N00014-2412297}.
Byeongjee Kang was supported in part by a fellowship from the
Ilju Academy and Culture Foundation.
Any opinions, findings, and conclusions or recommendations
in this material are those of the authors and do not necessarily
reflect the views of the funding agencies.

The authors thank the referees for their helpful feedback.
Claude Code (Fable 5.1) and Codex (GPT-5.6 Sol) were used in this work to
help fix grammar, proofread the text for clarity, develop the system
implementation and evaluation pipeline, and generate figures (R scripts)
and diagrams (\LaTeX{} TikZ).

\startappendix
\section{Artifact Appendix}
\label{sec:appendix-artifact}

\subsection{Abstract}

The artifact consists of two components:
(1) an energy modeling framework that measures the energy consumption of an
MSP430FR5994 microcontroller and derives its energy model (\S\ref{sec:parameter-estimation}), and
(2) a compilation framework, implemented as LLVM passes, that compiles a C program
into an MSP430 binary with region boundaries placed using the energy model from (1).
The compilation framework includes \tool and our reimplementations of the
baselines, \rockclimb and \schematic.
All components are open source.
We evaluate on the benchmarks included in the compilation framework repository,
most of which are taken from UlSWaP-Bench.

The two components can be run separately.
Reviewers can run (2) with our pre-measured energy model parameters without running (1).
In that case, correct execution under intermittent power is not guaranteed,
because parameters measured on one board may not match another board.
The full evaluation requires an MSP430 LaunchPad, measurement instruments,
and software licenses (see \S\ref{sec:ae-description}).
Without the hardware, reviewers can still run the compilation phase and obtain
instrumented binaries, provided they have a Gurobi license.

\subsection{Artifact check-list (meta-information)}

{\small
\begin{itemize}
  \item {\bf Program:}
  We use our annotated versions of benchmarks, mostly taken from UlSWaP-Bench.
  It's included in the artifact.
  \item {\bf Compilation:} Mostly LLVM-based pass. Included.
  \item {\bf Binary:} Binaries are not included, and users must build the artifact from source.
  \item {\bf Run-time environment:} macOS or Linux with Docker, uv, Saleae Logic 2, Gurobi, and Otii Automation Toolbox (the last two require licenses).
  \item {\bf Hardware:} MSP430FR5994 LaunchPad, Otii Ace Pro,
  Otii Switchboard, Saleae Logic 8,
  a Schottky diode, a \qty{560}{\ohm} resistor, and a macOS or Linux machine.
  \item {\bf Output:} Expected result included in the artifact.
  \item {\bf Experiments:} Python scripts (run with uv);
  compilation runs in a Docker container.
  \item {\bf How much disk space required (approximately)?:}
  About 21\,GB
  \item {\bf How much time is needed to complete experiments (approximately)?:}
  A few hours
  \item {\bf Publicly available?:} Yes
  \item {\bf Code licenses?:} MIT
  \item {\bf Archived (provide DOI)?:}
  Yes (\href{https://doi.org/10.5281/zenodo.22929039}{10.5281/zenodo.22929039},
  \href{https://doi.org/10.5281/zenodo.22931636}{10.5281/zenodo.22931636})
\end{itemize}
}

\subsection{Description}
\label{sec:ae-description}

\subsubsection{How to access}

Both components are publicly available on GitHub and archived on Zenodo:
\begin{itemize}
  \item Energy modeling framework:\\
        \url{https://github.com/byeongjee/msp430-energy-model}\\
        DOI: \href{https://doi.org/10.5281/zenodo.22929039}{10.5281/zenodo.22929039}
  \item Compilation framework:\\
        \url{https://github.com/byeongjee/bao}\\
        DOI: \href{https://doi.org/10.5281/zenodo.22931636}{10.5281/zenodo.22931636}
\end{itemize}
The Zenodo archives contain the evaluated versions,
and the GitHub repositories contain the latest versions.
Clone both repositories:
\begin{lstlisting}
git clone \
https://github.com/byeongjee/msp430-energy-model
git clone https://github.com/byeongjee/bao
\end{lstlisting}

\subsubsection{Hardware dependencies}

The full evaluation requires the following hardware:
\begin{itemize}
  \item MSP430FR5994 LaunchPad: the target device.
  \item Otii Ace Pro: measures power consumption.
  \item Otii Switchboard: automatically switches the LaunchPad between
        the debugger and the Otii Ace Pro.
  \item Saleae Logic 8: detects GPIO signals to measure execution time.
  \item Miscellaneous parts: a Schottky diode, a \qty{560}{\ohm} resistor, and cables.
\end{itemize}
Programs can be compiled on any x86-64 or ARM64 machine that can build LLVM%
\footnote{\url{https://llvm.org/docs/GettingStarted.html}}
and run Gurobi.%
\footnote{\url{https://support.gurobi.com/hc/en-us/articles/8172407217041-What-hardware-should-I-select-when-running-Gurobi}}
An M2 MacBook Pro is sufficient.

\subsubsection{Software dependencies}

The following software requires separate licenses:
\begin{itemize}
  \item Otii Automation Toolbox: controls the Otii Ace Pro from scripts.
  \item Gurobi: solves the MILP. Since the compiler runs in Docker,
        a Web License Service (WLS) license is required;
        the free academic WLS license suffices.
\end{itemize}
The following software is freely available but must be installed and started manually:
\begin{itemize}
  \item Saleae Logic 2:\footnote{\url{https://www.saleae.com/downloads}}
        our Python scripts control the Saleae Logic 8 through this software.
        Because it has no headless mode, it must be started and set up
        manually before running the scripts.
        See the README files in the repositories for details.
\end{itemize}
All other software (including msp430-gcc, mspdebug, and LLVM) is installed
by the Dockerfile in the compilation framework repository, which requires Docker.

\subsubsection{Data sets}

All data sets are in the compilation framework repository:
\begin{itemize}
  \item Benchmark programs: \lstinline{benchmarks/intermittent}.
  \item Energy-harvesting traces replayed in \S\ref{sec:intermittent-power}:\\
        \lstinline{benchmarks/traces}.
  \item Energy model parameters:\\
        \lstinline{benchmarks/assembly_params.json}.
        This is a sample measured on our board.
        Users are expected to generate parameters for their own board
        with the energy modeling framework.
\end{itemize}

\subsection{Installation}

Build the Docker image of the compilation framework and install the Python
dependencies with uv:
\begin{lstlisting}
cd bao
docker build -t bao .
export GRB_LICENSE_FILE=/path/to/gurobi.lic
export OTII_SERVER_BIN=/path/to/otii_server
uv sync --extra saleae --extra otii
\end{lstlisting}
For the energy modeling framework, run \lstinline{uv sync} in its repository.
Its Docker image is built on first use.
The wiring of the hardware and the setup of Saleae Logic 2 and Otii
are described in the README files of both repositories.

\subsection{Experiment Workflow}

The workflow consists of two steps:
\begin{enumerate}
  \item \emph{Energy modeling} (optional).
        Train the energy model parameters with \lstinline{uv run pem train}.
        The exact commands are given in the README of the energy modeling framework.
  \item \emph{Compilation and measurement}.
        Each experiment in the paper is a single \lstinline{uv run ckpt} command,
        listed in \S\ref{sec:ae-evaluation}.
\end{enumerate}

Readers without the hardware may run the following command
to compile all benchmarks without running them:
\begin{lstlisting}
uv run ckpt bench all --cap 5,10,50 \
  --save-build build
\end{lstlisting}
This still requires a Gurobi license.

\subsection{Evaluation and Expected Results}
\label{sec:ae-evaluation}

Each experiment in the paper can be reproduced with a single command
in the compilation framework repository.
Each command writes its results under \lstinline{results/},
where our raw results used in the paper are already stored.

The following experiments use direct wiring:
\begin{itemize}
  \item Execution time and region boundary hits (\S\ref{sec:continuous-power}, \S\ref{sec:milp-solving-time}):\\
        \lstinline{uv run ckpt bench all --cap 5,10,50}\\
        Results: \lstinline{milp.csv}, \lstinline{rockclimb.csv},
        \lstinline{schematic.csv},\\ \lstinline{schematicO3.csv},
        and \lstinline{uninstrumented.csv}.
  \item \rockclimb with unroll factor 64 on crc (\S\ref{sec:continuous-power}):\\
        \lstinline{uv run ckpt bench rockclimb-unroll64}\\
        Results: \lstinline{rockclimb_crc_unroll64.csv}.
  \item Execution time for the intermittent-power decomposition (\S\ref{sec:intermittent-power}):\\
        \lstinline{uv run ckpt bench decomposition}\\
        Results: \lstinline{intermittent/continuous/}.
  \item Loop chunking overhead (\S\ref{sec:chunking-overhead}):\\
        \lstinline{uv run ckpt bench chunking-overhead}\\
        Results: \lstinline{chunking_overhead/}.
  \item Trip count annotations (\S\ref{sec:trip-count}):\\
        \lstinline{uv run ckpt bench trip-count}\\
        Results: \lstinline{trip_count/}.
  \item Constant memory allocation (\S\ref{sec:speeding-up-milp}):\\
        \lstinline{uv run ckpt bench milp-coarse}\\
        Results: \lstinline{milp_coarse/}.
  \item Inlining (\S\ref{sec:evaluation-setup}, compile only, no board needed):\\
        \lstinline{uv run ckpt analyze inlining}\\
        Results: \lstinline{inlining.csv}.
\end{itemize}
The following experiment uses harvester wiring:
\begin{itemize}
  \item Execution under intermittent power (\S\ref{sec:intermittent-power}):\\
        \lstinline{uv run ckpt intermittent all}\\
        Results: \lstinline{intermittent/}.
\end{itemize}
The energy modeling framework stores the trained parameters in
\lstinline{benchmarks/assembly_params.json}.

\subsection{Notes}

Gurobi is deterministic on the same machine,
but its results may differ across machines.%
\footnote{\url{https://support.gurobi.com/hc/en-us/articles/360031636051-Is-Gurobi-deterministic}}
Therefore, the region-boundary placements of \tool may be slightly different from ours,
though the optimal objective value should remain the same.
Our results were obtained on an M2 MacBook Pro (32\,GB RAM) running macOS 26.6.2.

\bibliographystyle{ACM-Reference-Format}
\bibliography{references}

\AtEndDocument{%
\startsupplement
\section{Cost Semantics for Energy Accounting}
\label{sec:appendix-cost-semantics}

We provide a cost semantics for the energy consumption of program traces under a
given VM/NVM placement decision.

\paragraph{Events and traces.}
Let $\mathcal{A}$ be the set of events
\[
a ::= \mathsf{inst}(\iota)
\mid \mathsf{read}(\mu, v, n)
\mid \mathsf{write}(\mu, v, n),
\]
where $\iota$ ranges over MSP430 instruction classes other than the memory
read and write instructions \lstinline{mov} and \lstinline{memcpy}, which are
instead modeled by $\mathsf{read}$ and $\mathsf{write}$ events.
As described in \cref{sec:energy-profiling-assembly}, each instruction
class $\iota$ is a tuple
\[(\mathit{opcode}, \mathit{addrmode}_1, \dots, \mathit{addrmode}_n),\]
where $n \in \{0, 1, 2\}$ corresponds to nullary, unary, and binary
instructions, respectively.
For a memory access event,
$\mu \in \{\mathrm{VM}, \mathrm{NVM}\}$ indicates the type of memory that the
access targets, where $\mathrm{VM}$ denotes volatile memory and $\mathrm{NVM}$
denotes non-volatile memory;
$v \in \variables$ is the program variable being accessed; and
$n \in \mathbb{N}$ is the number of bytes accessed.
For example, an instruction \lstinline{mov &x, r5} that loads $n$ bytes of a
variable $x$ residing in NVM into a register is modeled by the event
$\mathsf{read}(\mathrm{NVM}, x, n)$, and \lstinline{mov r5, &x} where $x$
resides in VM is modeled by $\mathsf{write}(\mathrm{VM}, x, n)$.
Event traces are generated by the grammar
\[
\tau ::= \epsilon \mid a \cdot \tau,
\]
where $\epsilon$ is the empty trace and $\cdot$ denotes concatenation.
Let $\mathcal{T}$ be the set of traces generated by this grammar.
By convention, we also write $a$ for the singleton trace $a \cdot \epsilon$.
We use $\tau_1 \cdot \tau_2$ for concatenation of traces and omit the dot when
the meaning is clear.

\paragraph{Cost function.}
We define
\[
\operatorname{cost} : \mathcal{T} \to \mathbb{R}_{\ge 0}.
\]
The per-event cost $\operatorname{cost}(a)$ is given as an input parameter for
every event $a \in \mathcal{A}$, including both $\mathsf{inst}(\iota)$ events
and the memory access events $\mathsf{read}(\mu, v, n)$ and
$\mathsf{write}(\mu, v, n)$.
For instruction events, $\operatorname{cost}(\mathsf{inst}(\iota))$ comes from
the instruction-level model of
\cref{sec:energy-profiling-assembly}.
We further assume that the cost function has the following property:
there exist constants
$\deltaEread, \deltaEwrite \in \mathbb{R}_{\ge 0}$ such that for every
variable $v \in \variables$ and every $n \in \mathbb{N}$,
\begin{align*}
  \operatorname{cost}(\mathsf{read}(\mathrm{NVM}, v, n))
  &= \operatorname{cost}(\mathsf{read}(\mathrm{VM}, v, n)) + \deltaEread \cdot n, \\
  \operatorname{cost}(\mathsf{write}(\mathrm{NVM}, v, n))
  &= \operatorname{cost}(\mathsf{write}(\mathrm{VM}, v, n)) + \deltaEwrite \cdot n.
\end{align*}
Intuitively, the extra energy incurred by placing a variable in NVM rather
than VM is proportional to the number of bytes being read or written.
The constants $\deltaEread$ and $\deltaEwrite$ are also given as input
parameters.

We define $\operatorname{cost}$ recursively by
\[
\operatorname{cost}(\epsilon) \triangleq 0.
\]
\[
\operatorname{cost}(a \cdot \tau)
\triangleq
\operatorname{cost}(a) + \operatorname{cost}(\tau).
\]
Lemma~\ref{lem:cost-concat} and Lemma~\ref{lem:cost-finite-concat} record the
resulting additivity properties.

\begin{lemma}
\label{lem:cost-concat}
For any traces $\tau_1$ and $\tau_2$,
\[
\operatorname{cost}(\tau_1 \cdot \tau_2)
=
\operatorname{cost}(\tau_1) + \operatorname{cost}(\tau_2).
\]
\end{lemma}

\begin{lemma}
\label{lem:cost-finite-concat}
For any traces $\tau_1,\ldots,\tau_k$,
\[
\operatorname{cost}(\tau_1 \cdot \tau_2 \cdots \tau_k)
=
\sum_{i=1}^{k} \operatorname{cost}(\tau_i).
\]
\end{lemma}

For a finite sequence of traces $(\tau_i)_{i \in I}$, we write
\[
\prod_{i \in I} \tau_i
\]
for their concatenation in sequence order.
If $I = \varnothing$, this product is defined to be $\epsilon$.

\begin{definition}[Instrumentation templates]
Let $\tau_{\mathrm{pro}}$ be the fixed trace emitted by the region-start
runtime stub.
This trace captures only the region-start bookkeeping that the runtime
performs regardless of which variables are live, and is independent of the
MILP solution.
The per-variable restores that reload variables from NVM into VM after
recovering from a power failure are accounted for separately below, and are
not included in $\tau_{\mathrm{pro}}$.
Let $\tau_{\mathrm{epi}}$ be the fixed trace emitted by the region-end
runtime stub, which includes saving the program counter and stack pointer
to NVM so that execution can resume after a power failure.
As with $\tau_{\mathrm{pro}}$, the per-variable saves that commit live
variables back to NVM are accounted for separately below and are not
included in $\tau_{\mathrm{epi}}$.
For each variable $v \in \variables$,
let $\tau^{r}_{v}$ be the trace emitted to restore $v$ from NVM to VM, and let
$\tau^{s}_{v}$ be the trace emitted to save $v$ from VM to NVM.
We define
\[
\Epro \triangleq \operatorname{cost}(\tau_{\mathrm{pro}}),
\qquad
\Eepi \triangleq \operatorname{cost}(\tau_{\mathrm{epi}}),
\]
\[
\Erst_v \triangleq \operatorname{cost}(\tau^{r}_{v}),
\qquad
\Esv_v \triangleq \operatorname{cost}(\tau^{s}_{v}).
\]
$\Epro$ and $\Eepi$ are the fixed energy costs of the region-start and
region-end runtime stubs, while $\Erst_v$ and $\Esv_v$ are the per-variable
energy costs of restoring and saving $v$.
These four quantities are the energy parameters supplied as inputs to the
MILP solver in \cref{sec:appendix-algorithm}.
\end{definition}

\begin{definition}[All-VM block trace]
For each basic block $b \in B$,
let $\tau^{\mathrm{VM}}_{b}$ be the event trace of executing $b$ under the
reference placement in which every variable access uses VM.
We define
\[
\Ebase_b \triangleq \operatorname{cost}(\tau^{\mathrm{VM}}_{b}).
\]
\end{definition}

\begin{definition}[Bytes read and written]
For any trace $\tau$ and variable $v$,
let
\begin{align*}
\operatorname{bytesRead}_{\tau}(v)
&\triangleq
\sum_{\mathsf{read}(\mu, v, n) \in \tau} n,
\\
\operatorname{bytesWritten}_{\tau}(v)
&\triangleq
\sum_{\mathsf{write}(\mu, v, n) \in \tau} n.
\end{align*}
\end{definition}

We write
$\operatorname{vmOnly}_{\tau}(v)$ to mean that, for every $n \in \mathbb{N}$,
neither $\mathsf{read}(\mathrm{NVM}, v, n)$ nor
$\mathsf{write}(\mathrm{NVM}, v, n)$ occurs in $\tau$.

\begin{definition}[Relocating one variable]
For each variable $v \in \variables$,
define the relocation operator $\phi_v$ on events by
\begin{align*}
  \phi_v(\mathsf{read}(\mathrm{VM}, u, n))
  &=
  \begin{cases}
    \mathsf{read}(\mathrm{NVM}, u, n) & \text{if } u = v, \\
    \mathsf{read}(\mathrm{VM}, u, n) & \text{if } u \ne v,
  \end{cases} \\
  \phi_v(\mathsf{write}(\mathrm{VM}, u, n))
  &=
  \begin{cases}
    \mathsf{write}(\mathrm{NVM}, u, n) & \text{if } u = v, \\
    \mathsf{write}(\mathrm{VM}, u, n) & \text{if } u \ne v,
  \end{cases}
\end{align*}
and let $\phi_v$ act as the identity on all other events.
We extend $\phi_v$ to traces by structural recursion:
\[
\phi_v(\epsilon) \triangleq \epsilon.
\]
\[
\phi_v(a \cdot \tau)
\triangleq
\phi_v(a) \cdot \phi_v(\tau).
\]
\end{definition}
Intuitively, $\phi_v$ changes exactly the VM accesses to $v$ into NVM
accesses, leaving every other event unchanged.

\begin{lemma}
\label{lem:single-relocation-cost}
Let $\tau$ be a trace and let $v \in \variables$.
If $\operatorname{vmOnly}_{\tau}(v)$, then
\begin{align*}
\operatorname{cost}(\phi_v(\tau))
&=
\operatorname{cost}(\tau) \\
&\quad
+ \deltaEread \cdot \operatorname{bytesRead}_{\tau}(v) \\
&\quad
+ \deltaEwrite \cdot \operatorname{bytesWritten}_{\tau}(v).
\end{align*}
\end{lemma}
\begin{proof}
By induction on $\tau$.
If $\tau = \epsilon$, then $\phi_v(\tau)=\epsilon$ and
\begin{align*}
\operatorname{bytesRead}_{\epsilon}(v) &= 0, \\
\operatorname{bytesWritten}_{\epsilon}(v) &= 0,
\end{align*}
so the claim is immediate.
Now let $\tau = a \cdot \tau'$ and assume
$\operatorname{vmOnly}_{a \cdot \tau'}(v)$.
Then $\operatorname{vmOnly}_{\tau'}(v)$, so by the induction hypothesis,
\begin{align*}
\operatorname{cost}(\phi_v(\tau'))
&=
\operatorname{cost}(\tau') \\
&\quad
+ \deltaEread \cdot \operatorname{bytesRead}_{\tau'}(v) \\
&\quad
+ \deltaEwrite \cdot \operatorname{bytesWritten}_{\tau'}(v).
\end{align*}
We distinguish three cases.
If $a$ does not access $v$, then $\phi_v(a)=a$ and
\begin{align*}
\operatorname{bytesRead}_{a \cdot \tau'}(v)
&=
\operatorname{bytesRead}_{\tau'}(v), \\
\operatorname{bytesWritten}_{a \cdot \tau'}(v)
&=
\operatorname{bytesWritten}_{\tau'}(v).
\end{align*}
Hence
\begin{align*}
  \operatorname{cost}(\phi_v(a \cdot \tau'))
  &= \operatorname{cost}(\phi_v(a)) + \operatorname{cost}(\phi_v(\tau')) \\
  &= \operatorname{cost}(a \cdot \tau') \\
  &\quad + \deltaEread \cdot \operatorname{bytesRead}_{a \cdot \tau'}(v) \\
  &\quad + \deltaEwrite \cdot \operatorname{bytesWritten}_{a \cdot \tau'}(v).
\end{align*}
If $a = \mathsf{read}(\mathrm{VM}, v, n)$, then
$\phi_v(a)=\mathsf{read}(\mathrm{NVM}, v, n)$ and
\[
\operatorname{cost}(\phi_v(a))=\operatorname{cost}(a)+\deltaEread \cdot n.
\]
Moreover,
\begin{align*}
\operatorname{bytesRead}_{a \cdot \tau'}(v)
&=
n+\operatorname{bytesRead}_{\tau'}(v), \\
\operatorname{bytesWritten}_{a \cdot \tau'}(v)
&=
\operatorname{bytesWritten}_{\tau'}(v).
\end{align*}
Therefore
\begin{align*}
  \operatorname{cost}(\phi_v(a \cdot \tau'))
  &= \operatorname{cost}(\phi_v(a)) + \operatorname{cost}(\phi_v(\tau')) \\
  &= \operatorname{cost}(a \cdot \tau') \\
  &\quad + \deltaEread \cdot \operatorname{bytesRead}_{a \cdot \tau'}(v) \\
  &\quad + \deltaEwrite \cdot \operatorname{bytesWritten}_{a \cdot \tau'}(v).
\end{align*}
If $a = \mathsf{write}(\mathrm{VM}, v, n)$, the argument is analogous:
$\phi_v(a)=\mathsf{write}(\mathrm{NVM}, v, n)$,
\[
\operatorname{cost}(\phi_v(a))=\operatorname{cost}(a)+\deltaEwrite \cdot n,
\]
and
\begin{align*}
\operatorname{bytesRead}_{a \cdot \tau'}(v)
&=
\operatorname{bytesRead}_{\tau'}(v), \\
\operatorname{bytesWritten}_{a \cdot \tau'}(v)
&=
n+\operatorname{bytesWritten}_{\tau'}(v).
\end{align*}
Substituting these equalities yields the claim.
These are the only possibilities because $\operatorname{vmOnly}_{a \cdot \tau'}(v)$
rules out accesses of the form $\mathsf{read}(\mathrm{NVM}, v, n)$ and
$\mathsf{write}(\mathrm{NVM}, v, n)$.
\end{proof}

\begin{definition}[Relocating a set of variables]
Fix any total order $\prec$ on $\variables$ once and for all.
For a finite set $S = \{v_1,\ldots,v_m\} \subseteq \variables$ with
$v_1 \prec \cdots \prec v_m$,
define
\[
\phi_S \triangleq \phi_{v_m} \circ \cdots \circ \phi_{v_1}.
\]
Thus, $\phi_S$ is obtained by relocating one variable at a time.
For a finite set $S \subseteq \variables$, write
\[
\operatorname{vmOnly}_{\tau}(S)
\triangleq
\forall v \in S.\ \operatorname{vmOnly}_{\tau}(v).
\]
\end{definition}

\begin{lemma}
\label{lem:set-relocation-cost}
Let $\tau$ be a trace and let $S \subseteq \variables$ be finite.
If $\operatorname{vmOnly}_{\tau}(S)$, then
\begin{align*}
\operatorname{cost}(\phi_S(\tau))
&=
\operatorname{cost}(\tau) \\
&\quad
+ \sum_{v \in S} \deltaEread \cdot \operatorname{bytesRead}_{\tau}(v) \\
&\quad
+ \sum_{v \in S} \deltaEwrite \cdot \operatorname{bytesWritten}_{\tau}(v).
\end{align*}
\end{lemma}
\begin{proof}
Let $\tau_0 \triangleq \tau$ and, for $i=1,\ldots,m$, let
\[
\tau_i \triangleq \phi_{v_i}(\tau_{i-1}).
\]
Then $\tau_m = \phi_S(\tau)$.
Fix $i \in \{1,\ldots,m\}$.
Since $\operatorname{vmOnly}_{\tau}(S)$ holds and each operator
$\phi_{v_j}$ with $j<i$ changes only accesses to $v_j$, we have
\begin{align*}
\operatorname{vmOnly}_{\tau_{i-1}}(v_i), \\
\operatorname{bytesRead}_{\tau_{i-1}}(v_i)
&=
\operatorname{bytesRead}_{\tau}(v_i), \\
\operatorname{bytesWritten}_{\tau_{i-1}}(v_i)
&=
\operatorname{bytesWritten}_{\tau}(v_i).
\end{align*}
Applying Lemma~\ref{lem:single-relocation-cost} to $(\tau_{i-1}, v_i)$ yields
\begin{align*}
\operatorname{cost}(\tau_i)
\;=\;
\operatorname{cost}(\tau_{i-1}) \\
&\quad
+ \deltaEread \cdot \operatorname{bytesRead}_{\tau}(v_i) \\
&\quad
+ \deltaEwrite \cdot \operatorname{bytesWritten}_{\tau}(v_i).
\end{align*}
Summing these equalities over \(i=1,\ldots,m\) gives the result.
\end{proof}

\begin{definition}[MILP block cost]
We define the refined per-variable NVM penalty by
\begin{align*}
\deltaE_{b,v}
&\triangleq
\deltaEread \cdot \operatorname{bytesRead}_{\tau^{\mathrm{VM}}_{b}}(v) \\
&\quad
+ \deltaEwrite \cdot \operatorname{bytesWritten}_{\tau^{\mathrm{VM}}_{b}}(v).
\end{align*}
This is the exact extra model cost incurred by relocating all accesses to $v$
in block $b$ from VM to NVM.
Under the simplifying assumption that every access to $v$ reads or writes the
entire object,
this reduces to the shorthand used in \cref{sec:appendix-algorithm}:
\begin{align*}
\deltaE_{b,v}
&=
\deltaEread \cdot \size(v) \cdot \numReads(b,v) \\
&\quad
+ \deltaEwrite \cdot \size(v) \cdot \numWrites(b,v).
\end{align*}

Given a feasible MILP solution $\sigma$,
let
\[
S_b(\sigma) \triangleq \{v \in \variables \mid m_{b,v}(\sigma)=0\}
\]
be the set of variables placed in NVM while executing $b$.
Define the execution trace of block $b$ under $\sigma$ by
\[
\tau^{\mathrm{blk}}_{b}(\sigma)
\triangleq
\phi_{S_b(\sigma)}(\tau^{\mathrm{VM}}_{b}).
\]
\end{definition}

\begin{lemma}
\label{lem:block-cost-relocation}
For every feasible MILP solution $\sigma$ and every block $b \in B$,
\begin{align*}
\operatorname{cost}(\tau^{\mathrm{blk}}_{b}(\sigma))
&=
\Ebase_b + \sum_{v \in \variables} \deltaE_{b,v}(1-m_{b,v}(\sigma)) \\
&=
\eblk_b(\sigma).
\end{align*}
\end{lemma}
\begin{proof}
By definition of $\tau^{\mathrm{VM}}_{b}$, every variable access in
$\tau^{\mathrm{VM}}_{b}$ uses VM.
Hence Lemma~\ref{lem:set-relocation-cost} applies, giving
\begin{align*}
\operatorname{cost}(\tau^{\mathrm{blk}}_{b}(\sigma))
&=
\operatorname{cost}(\tau^{\mathrm{VM}}_{b}) + \sum_{v \in S_b(\sigma)} \deltaE_{b,v}.
\end{align*}
By the definitions of $\Ebase_b$, $\deltaE_{b,v}$, and $S_b(\sigma)$, this
becomes
\[
\Ebase_b
+
\sum_{v \in \variables} \deltaE_{b,v}(1-m_{b,v}(\sigma))
=
\eblk_b(\sigma).
\]
\end{proof}

\begin{definition}[Region-boundary traces]
For a block $b$,
define the region-start trace
\[
\tau^{+}_{b}(\sigma)
\triangleq
\begin{cases}
\tau_{\mathrm{pro}}
\cdot
\displaystyle\prod_{\substack{v \in \variables \\ \hat{r}_{b,v}(\sigma)=1}}
\tau^{r}_{v}
& \text{if } r_b(\sigma)=1,\\[0.8em]
\epsilon
& \text{if } r_b(\sigma)=0,
\end{cases}
\]
and the region-end trace
\[
\tau^{-}_{b}(\sigma)
\triangleq
\begin{cases}
\displaystyle\prod_{\substack{v \in \variables \\ s_{b,v}(\sigma)=1}}
\tau^{s}_{v}
\cdot
\tau_{\mathrm{epi}}
& \text{if } r_b(\sigma)=1,\\[0.8em]
\epsilon
& \text{if } r_b(\sigma)=0.
\end{cases}
\]
\end{definition}
This exactly matches with our instrumentation in
\cref{sec:instrumentation}.

\begin{lemma}
\label{lem:boundary-template-cost}
For every feasible MILP solution $\sigma$ and every block $b \in B$,
\begin{align*}
\operatorname{cost}(\tau^{+}_{b}(\sigma)) &= \estart_b(\sigma), \\
\operatorname{cost}(\tau^{-}_{b}(\sigma)) &= \eend_b(\sigma).
\end{align*}
where $\estart_b(\sigma)$ and $\eend_b(\sigma)$ are defined in
\cref{sec:appendix-algorithm}.
\end{lemma}
\begin{proof}
If $r_b(\sigma)=0$, then
$\tau^{+}_{b}(\sigma)=\tau^{-}_{b}(\sigma)=\epsilon$,
and both $\estart_b(\sigma)$ and $\eend_b(\sigma)$ are zero.
If $r_b(\sigma)=1$, then Lemma~\ref{lem:cost-finite-concat} gives
\begin{align*}
\operatorname{cost}(\tau^{+}_{b}(\sigma))
&=
\operatorname{cost}(\tau_{\mathrm{pro}})
  + \sum_{\substack{v \in \variables \\ \hat{r}_{b,v}(\sigma)=1}}
    \operatorname{cost}(\tau^{r}_{v}) \\
&=
\Epro + \sum_{v \in \variables} \Erst_v \cdot \hat{r}_{b,v}(\sigma) \\
&=
\estart_b(\sigma),
\end{align*}
and similarly
\begin{align*}
\operatorname{cost}(\tau^{-}_{b}(\sigma))
&=
\sum_{\substack{v \in \variables \\ s_{b,v}(\sigma)=1}}
\operatorname{cost}(\tau^{s}_{v})
  + \operatorname{cost}(\tau_{\mathrm{epi}}) \\
&=
\sum_{v \in \variables} \Esv_v \cdot s_{b,v}(\sigma) + \Eepi \\
&=
\eend_b(\sigma).
\end{align*}
\end{proof}

\begin{definition}[Path semantics]
Let $\pi = b_0 \to b_1 \to \cdots \to b_k$ be a CFG path such that
$r_{b_0}(\sigma)=1$ and $r_{b_i}(\sigma)=0$ for all $1 \le i \le k$.
Define its execution trace under $\sigma$ by
\[
\tau_{\pi,\sigma}
\triangleq
\tau^{+}_{b_0}(\sigma)
\cdot
\tau^{\mathrm{blk}}_{b_0}(\sigma)
\cdot
\tau^{\mathrm{blk}}_{b_1}(\sigma)
\cdots
\tau^{\mathrm{blk}}_{b_k}(\sigma).
\]
We then define
\[
\operatorname{cost}(\pi(\sigma))
\triangleq
\operatorname{cost}(\tau_{\pi,\sigma}).
\]

If $\pi' = b_0 \to b_1 \to \cdots \to b_k \to b_{k+1}$ and
$r_{b_{k+1}}(\sigma)=1$,
define
\[
\tau_{\pi',\sigma}
\triangleq
\tau^{+}_{b_0}(\sigma)
\cdot
\tau^{\mathrm{blk}}_{b_0}(\sigma)
\cdot
\tau^{\mathrm{blk}}_{b_1}(\sigma)
\cdots
\tau^{\mathrm{blk}}_{b_k}(\sigma)
\cdot
\tau^{-}_{b_{k+1}}(\sigma),
\]
and again let
$\operatorname{cost}(\pi'(\sigma)) \triangleq \operatorname{cost}(\tau_{\pi',\sigma})$.
\end{definition}

\getkeytheorem{energyEqualityCostSemantics}
\begin{proof}
By definition of $\tau_{\pi,\sigma}$ and Lemma~\ref{lem:cost-finite-concat},
\[
\operatorname{cost}(\pi(\sigma))
=
\operatorname{cost}(\tau^{+}_{b_0}(\sigma))
+
\sum_{i=0}^{k}
\operatorname{cost}(\tau^{\mathrm{blk}}_{b_i}(\sigma)).
\]
Applying Lemma~\ref{lem:boundary-template-cost} and
Lemma~\ref{lem:block-cost-relocation} yields
\[
\operatorname{cost}(\pi(\sigma))
=
\estart_{b_0}(\sigma)
+
\sum_{i=0}^{k} \eblk_{b_i}(\sigma).
\]
The second equality is identical, except that
$\tau_{\pi',\sigma}$ contains the additional suffix
$\tau^{-}_{b_{k+1}}(\sigma)$, whose cost is
$\eend_{b_{k+1}}(\sigma)$ by Lemma~\ref{lem:boundary-template-cost}.
\end{proof}

\section{MILP Formulation}
\label{sec:appendix-algorithm}

\captionsetup[subtable]{skip=0pt}

\subsection{Definitions}

We assume that a program $P$ is
represented as a
control flow graph $G_P=(B,E)$ defined at the level of SSA-based IR,
where $B$ is the set of basic blocks and $E\subseteq B\times B$ is the set of directed edges.
We assume that $b_0\in B$ is the unique entry block
and $B_{\mathrm{exit}}\subseteq B$ is the set of exit blocks.
We use $\variables$ to denote the set of program variables,
including global variables, stack slots (e.g., \lstinline{alloca}),
and virtual registers.
We do not support heap allocations,
which are rare in energy-constrained embedded systems.

The set of program variables is partitioned into the following sets:
$\variables = \eligibles \cup \ineligibles$.
$\eligibles$ is the set of variables
that \tool can determine their allocation sites (VM or NVM),
and $\ineligibles$ is the set of variables that
are always in VM and not considered for placement.

\begin{table}[t]
\centering
\caption{Notation used in the MILP formulation.}
\Description{Table listing hardware and variable attribute symbols used in the MILP formulation.}
\label{fig:notation}
\begin{subtable}[t]{\linewidth}
\centering
\caption{Memory analysis}
\label{fig:notation-memory}
\begin{tabular}{@{}p{0.25\linewidth}p{0.7\linewidth}@{}}
\toprule
\textbf{Symbol} & \textbf{Description} \\
\midrule
$\LiveIn(b) \subseteq \variables$
  & Variables that are live at the entry of $b$. \\
$\defs(b) \subseteq \variables$
  & Variables that may be defined in $b$:
    $v \in \defs(b)$ if $v$ is a virtual register defined in $b$,
    or a stack slot or global variable written to in $b$. \\
$D_{b,v} \in \{0,1\}$
  & Definition indicator: $1$ if $v \in \defs(b)$, $0$ otherwise. \\
$L_{b,v} \in \{0,1\}$
  & Live indicator: $1$ if $v \in \LiveIn(b)$, $0$ otherwise. \\
\bottomrule
\end{tabular}
\end{subtable}

\vspace{0.5em}

\begin{subtable}[t]{\linewidth}
\centering
\caption{Hardware and variable attributes}
\label{fig:notation-hw}
\begin{tabular}{@{}p{0.25\linewidth}p{0.7\linewidth}@{}}
\toprule
\textbf{Symbol} & \textbf{Description} \\
\midrule
$\Svm$
  & Usable VM capacity in bytes. \\
$\size(v)$
  & Size of $v$ in bytes. \\
\bottomrule
\end{tabular}
\end{subtable}
\end{table}

\begin{table}[t]
\centering
\caption{Input parameters for the MILP formulation.}
\Description{Tables listing structural and profiling input parameters for the MILP formulation.}
\label{fig:params}
\begin{subtable}[t]{\linewidth}
\centering
\caption{Energy parameters}
\label{fig:params-energy}
\begin{tabular}{@{}p{0.25\linewidth}p{0.7\linewidth}@{}}
\toprule
\textbf{Symbol} & \textbf{Description} \\
\midrule
$\ebuffer$
  & Buffer energy budget. \\
$\Ebase_b$
  & Energy to execute block $b$, assuming all $v \in \eligibles$ are placed in VM. \\
$\deltaEread, \deltaEwrite$
  & Extra energy per byte for a read/write to NVM rather than VM. \\
$\Epro$
  & Region-start energy (fixed overhead, excluding restores). \\
$\Eepi$
  & Region-end energy (fixed overhead, excluding saves). \\
$\Eread$
  & Energy per byte to restore a variable from NVM to VM. \\
$\Ewrite$
  & Energy per byte to save a variable from VM to NVM. \\
$\Erst_v$ & Energy to restore $v$ from NVM to VM \newline
   ($\Eread \cdot \size(v)$) \\
$\Esv_v  $
    & Energy to save $v$ from VM to NVM \newline
    ($\Ewrite \cdot \size(v)$) \\
\bottomrule
\end{tabular}
\end{subtable}

\vspace{0.5em}

\begin{subtable}[t]{\linewidth}
\centering
\caption{Profiling data, obtained by profiling, symbolic execution, or other means}
\label{fig:params-profile}
\begin{tabular}{@{}p{0.25\linewidth}p{0.7\linewidth}@{}}
\toprule
\textbf{Symbol} & \textbf{Description} \\
\midrule
$F_b$
  & Expected dynamic number of times block $b$ is entered. \\
\bottomrule
\end{tabular}
\end{subtable}
\end{table}

\subsection{MILP Variables}

\begin{table*}[t]
\centering
\caption{MILP variables, which are comprised of decision variables and derived variables.}
\Description{Tables listing MILP decision variables and derived variables with their descriptions and definitions.}
\label{fig:milp-vars}
\renewcommand{\arraystretch}{1.3}
\begin{subtable}[t]{\linewidth}
\centering
\caption{Decision variables}
\label{fig:milp-decision}
\begin{tabular}{@{}p{0.15\linewidth}p{0.8\linewidth}@{}}
\toprule
\textbf{Symbol} & \textbf{Description} \\
\midrule
$r_b \in \{0,1\}$
  & $1$ iff a region boundary is placed at the start of block $b$. \\
$m_{b,v} \in \{0,1\}$
  & $1$ iff $v$ is placed in VM while executing $b$. \\
$\hat{r}_{b,v} \in \{0,1\}$
  &  $r_b \wedge m_{b,v} \wedge L_{b,v}$.
  If $\hat{r}_{b,v} = 1$, then $v$ is live-in, placed in VM, and a region starts
  at $b$, so $v$ must be restored from NVM to VM at the region start.
  \\
$d_{b,v} \in \{0,1\}$
  & $1$ iff $v$ is dirty in the current region containing $b$,
    i.e., $v \in \defs(b')$ for some $b'$ on a path
    from the last region boundary to $b$. \\
$\hat{d}_{b,v} \in \{0,1\}$
  & Linearization of $d_{b,v} \wedge m_{b,v}$:
    $v$ is dirty and placed in VM at $b$. \\
$s_{b,v} \in \{0,1\}$
  & $1$ iff $v$ must be saved at a region boundary preceding $b$. \\
$\eaccum_b \in [0, \ebuffer]$
  & Worst-case energy consumed from the start of the current region
    up to (just before) executing $b$. \\
\bottomrule
\end{tabular}
\end{subtable}

\vspace{0.5em}

\begin{subtable}[t]{\linewidth}
\centering
\caption{Derived variables}
\label{fig:milp-derived}
\begin{tabular}{@{}p{0.15\linewidth}p{0.3\linewidth}p{0.5\linewidth}@{}}
\toprule
\textbf{Symbol} & \textbf{Description} & \textbf{Definition} \\
\midrule
$\deltaE_{b,v}$
  & Extra energy in $b$ from placing $v$ in NVM
  & $\deltaEread \cdot \size(v) \cdot \numReads(b,v) + \deltaEwrite \cdot \size(v) \cdot \numWrites(b,v)$ \\
$\eblk_b$
  & Energy to execute block $b$
  & $\Ebase_b + \sum_{v \in \variables} \deltaE_{b,v}(1-m_{b,v})$ \\
$\estart_b$
  & Energy to start a region at $b$ \newline (fixed region-start cost + restoring live-in variables)
  & $\Epro \cdot r_b + \sum_{v \in \variables} \Erst_v \cdot \hat{r}_{b,v}$ \\
$\eend_b$
  & Energy to end a region before $b$ \newline (fixed region-end cost + saving dirty variables)
  & $\Eepi \cdot r_b + \sum_{v \in \variables} \Esv_v \cdot s_{b,v}$ \\
\bottomrule
\end{tabular}
\end{subtable}
\end{table*}

\subsection{Objective Function}

The objective function to minimize is given as follows.
We minimize the expected energy overhead incurred by memory placement and region boundary code.

\begin{equation*}
  \sum_{b\in B} F_b \cdot (\eblk_b + \estart_b + \eend_b)
\end{equation*}

\subsection{Constraints}

\begin{constraint}[Entry block is a region boundary.]
  \begin{equation*}
    r_{b_0} = 1
  \end{equation*}
\end{constraint}

\begin{constraint}[VM capacity constraint.]
  We cannot allocate more variables in VM than the VM capacity allows
  \begin{equation*}
    \sum_{v \in \eligibles} \size(v) \cdot m_{b,v}
    + \sum_{v \in \ineligibles} \size(v)
    \le \Svm
    \tag{$\forall b \in B$}
  \end{equation*}
\end{constraint}

\begin{constraint}
  \label{c:ineligible-always-in-vm}
  All ineligible variables are unconditionally placed in VM.
  \begin{equation*}
    m_{b,v} = 1 \tag{$\forall b \in B, \forall v\in \ineligibles$}
  \end{equation*}
\end{constraint}

\begin{constraint}
  Encoding of $\hat{r}_{b,v}$ as a logical conjunction of $r_b$ and $m_{b,v}$.
  \begin{equation*}
    \hat{r}_{b,v} = r_b \wedge m_{b,v} \wedge L_{b,v} \tag{$\forall b\in B, \forall v\in \variables$}
  \end{equation*}
\end{constraint}

\begin{constraint}
  Variables defined in a block are dirty.
  \begin{equation*}
    d_{b,v} \ge D_{b,v} \tag{$\forall b \in B, \forall v\in \variables$}
  \end{equation*}
  \end{constraint}

\begin{constraint}
  \label{c:dirtiness-propagation}
  Dirtiness propagates if we do not place a region boundary in between.
  \begin{equation*}
   d_{\succ, v} \ge d_{\pred, v}
    \quad \text{if } r_{\succ} = 0
   \tag{$\forall (\pred, \succ)\in E, \forall v\in \variables$}
  \end{equation*}

\end{constraint}

\begin{constraint}
  Region boundary resets dirty.
  \begin{equation*}
    d_{b,v} = D_{b,v}
    \quad \text{if } r_b = 1
    \tag{$\forall b \in B, \forall v\in \variables$}
  \end{equation*}
  In implementation, we use the following constraint instead:
  \begin{align*}
    d_{b,v} &\ge D_{b,v} + r_b - 1 \\
    d_{b,v} &\le D_{b,v} - r_b + 1 \\
    \tag{$\forall b \in B, \forall v\in \variables$}
  \end{align*}
  We also add the following LP tightening cut:
  \begin{equation*}
    d_{\succ,v} \le D_{\succ,v} + \sum_{(\pred, \succ)\in E} d_{\pred, v}
    \tag{$\forall (\pred, \succ)\in E, \forall v\in \variables$}
  \end{equation*}
  This is redundant for integer solutions
  but tightens the LP relaxation:
  $d_{\succ,v}$ can be $1$ only if $v$ is defined in $\succ$
  or some predecessor carries $d_{\pred,v} = 1$.
\end{constraint}

\begin{constraint}
  Encoding of $\hat{d}_{b,v}$ as a logical conjunction of $d_{b,v}$ and $m_{b,v}$.
  \begin{equation*}
    \hat{d}_{b,v} = d_{b,v} \wedge m_{b,v} \tag{$\forall b \in B, \forall v\in \variables$}
  \end{equation*}

\end{constraint}

\begin{constraint}
  All dirty variables placed in VM are saved at a region boundary.
  \begin{equation*}
    s_{\succ,v} = 1
    \quad \text{if } \hat{d}_{\pred,v} = 1 \wedge r_{\succ} = 1
    \tag{$\forall (\pred, \succ)\in E, \forall v\in \variables$}
  \end{equation*}
  In implementation, we use the following constraint instead:
  \begin{align*}
    s_{\succ,v} \ge \hat{d}_{\pred,v} + r_{\succ} - 1
    \tag{$\forall (\pred, \succ)\in E, \forall v\in \variables$}
  \end{align*}
  We also add the following LP tightening cuts:
  \begin{align*}
    s_{\succ,v} &\le r_{\succ} \\
    s_{\succ,v} &\le \sum_{(\pred, \succ)\in E} \hat{d}_{\pred, v}
    \tag{$\forall b \in B, \forall v\in \variables$}
  \end{align*}
  These are redundant for integer solutions
  but tighten the LP relaxation:
  $s_{\succ,v}$ can be $1$ only at a region boundary
  where some predecessor has $\hat{d}_{\pred,v} = 1$.
  Additionally, $s_{b,v}$ is only needed for $v \in \LiveIn(b)$,
  since dead variables do not require saving.
\end{constraint}

\begin{constraint}
  \label{c:accum-init}
  Energy accumulated at a region boundary is $\estart_b$.
  \begin{equation*}
    \eaccum_b = \estart_b
    \quad \text{if } r_b = 1
    \tag{$\forall b \in B$}
  \end{equation*}
\end{constraint}

\begin{constraint}
  \label{c:accum-prop}
  Accumulated energy propagates within a region.
  \begin{equation*}
    \eaccum_{\succ} \ge \eaccum_{\pred} + \eblk_{\pred}
    \quad \text{if } r_{\succ} = 0
    \tag{$\forall (\pred, \succ)\in E$}
  \end{equation*}
\end{constraint}

\begin{constraint}
  \label{c:energy-bound}
  For every edge $(\pred, \succ)$, the region has enough energy to execute
  $\pred$ and, if $\succ$ is a region start, also pay the
  cost of the region-end and saving all dirty variables.
  \begin{equation*}
    \eaccum_\pred +
    \eblk_\pred +
    \eend_\succ
    \le \ebuffer \tag{$\forall (\pred, \succ)\in E$}
  \end{equation*}
\end{constraint}

\begin{constraint}[Energy budget constraint for exit blocks.]
  \label{c:energy-bound-exit}
  For every exit block $b \in B_{\mathrm{exit}}$,
  there should be enough energy to finish $b$.
  \begin{equation*}
    \eaccum_b + \eblk_b \le \ebuffer \tag{$\forall b \in B_{\mathrm{exit}}$}
  \end{equation*}
\end{constraint}

\begin{constraint}
  \label{c:placement-propagation}
  A variable's allocation site can change only at region boundaries.
  \begin{equation*}
    m_{\pred,v} = m_{\succ,v}
    \quad \text{if } r_{\succ} = 0
    \tag{$\forall (\pred, \succ)\in E, \forall v\in \variables$}
  \end{equation*}
  In implementation, we use the following constraint instead:
  \begin{align*}
    m_{\pred,v} &\ge m_{\succ,v} - r_{\succ}  \\
    m_{\pred,v} &\le m_{\succ,v} + r_{\succ}
    \tag{$\forall (\pred, \succ)\in E, \forall v\in \variables$}
  \end{align*}
\end{constraint}

\begin{constraint}[Allocation agreement at join points.]
  \label{c:allocation-join}
  All predecessors of a block agree on where each variable is placed.
  \begin{equation*}
    m_{u,v}=m_{w,v}
    \tag{$\forall b\in B, \forall (u,b),(w,b)\in E,
      \forall v\in\eligibles$}
  \end{equation*}
  This is redundant when $r_b = 0$, since \cmark{c:placement-propagation}
  already equates every predecessor of $b$ with $b$ itself.
  It is required when $r_b = 1$, because the region end at $b$ saves each
  variable once for the block rather than once per incoming edge:
  if $v$ were placed in VM at $u$ and in NVM at $w$,
  the save forced by $\hat{d}_{u,v} = 1$ would also be executed on the path
  through $w$, overwriting the up-to-date NVM value with a stale VM copy.
  Ineligible variables need no equation, since
  \cmark{c:ineligible-always-in-vm} places them in VM at every block.

  In implementation, we do not emit one equation per pair of predecessors.
  Writing $u_1, \ldots, u_n$ for the predecessors of $b$,
  we equate each predecessor with the first,
  which is $n-1$ equations rather than $\binom{n}{2}$:
  \begin{equation*}
    m_{u_1,v} = m_{u_i,v}
    \tag{$\forall b\in B, 2\le i\le n, \forall v\in\eligibles$}
  \end{equation*}
\end{constraint}

\subsection{Safety Theorem}

\begin{definition}
We use $\sigma$ to denote an assignment to all MILP variables,
i.e., $\{r_b\}_{b\in B}$, $\{m_{b,v}\}_{b\in B, v\in \variables}$,
$\{d_{b,v}\}$, $\{\hat{d}_{b,v}\}$, $\{s_{b,v}\}$, and $\{\eaccum_b\}$
(indexed over all relevant $b$ and $v$).
$\sigma$ is \emph{feasible} if it satisfies all constraints.
It is \emph{optimal} if it is feasible and minimizes the objective function.
\end{definition}

\begin{definition}
For an MILP variable $x$, we use $x(\sigma)$
to denote the value of $x$ in the solution $\sigma$.
Similarly, for derived variables,
such as $\eblk_b$,
we use $\eblk_b(\sigma)$ to denote the value of $\eblk_b$ in the solution $\sigma$.
\end{definition}

\begin{definition}
\label{def:instrumented-program}
  Given a feasible MILP solution $\sigma$ for a program $P$,
  the instrumented program $P(\sigma)$
  is obtained by transforming $P$ as follows:
  \begin{description}
    \item[Memory Placement]
    For each $v \in \variables$,
    $v$ is allocated in VM at block $b$ iff $m_{b,v}(\sigma) = 1$.
    \item[Region Boundary]
    At each block $b$ with $r_b(\sigma) = 1$, we insert
    \begin{enumerate}
      \item Region end that saves every variable $v$ with $s_{b,v}(\sigma) = 1$
      to NVM and performs a fixed-cost action (e.g., power off).
      \item Region start that
      performs a fixed-cost action (e.g., power on) and then
      restores every variable $v$ with $\hat{r}_{b,v}(\sigma) = 1$ from NVM to VM.
    \end{enumerate}
  \end{description}
\end{definition}

Informally,
we say that $P(\sigma)$ is \emph{safe} if
\begin{enumerate*}[label=(\arabic*)]
\item the energy consumed within each region does not exceed the budget $\ebuffer$, and
\item all VM variables are saved and restored correctly at region boundaries,
\end{enumerate*}
ensuring memory consistency across power failures.
Our goal is to show that $P(\sigma)$ is safe
whenever $\sigma$ is a feasible MILP solution.

We say that a CFG path $\pi = b_0 \to b_1 \to \cdots \to b_k$ stays within the
same region if $r_{b_0}(\sigma) = 1$ and $r_{b_i}(\sigma) = 0$ for all
$1 \le i \le k$.

We say that a CFG path $\pi$ reaches the next region boundary if
$\pi = b_0 \to b_1 \to \cdots \to b_k \to b_{k+1}$,
$r_{b_0}(\sigma) = r_{b_{k+1}}(\sigma) = 1$, and
$r_{b_i}(\sigma) = 0$ for all $1 \le i \le k$.

We say that a CFG path $\pi$ reaches program exit without crossing another
region boundary if $\pi = b_0 \to b_1 \to \cdots \to b_k$,
$r_{b_0}(\sigma) = 1$, and
$r_{b_i}(\sigma) = 0$ for all $1 \le i \le k$.

\begin{lemma}
  \label{lem:accum-init}
  Let $\sigma$ be a feasible MILP solution.
  For all $b \in B$, if $r_b(\sigma) = 1$, then
  $\eaccum_b(\sigma) = \estart_b(\sigma)$ and $\estart_b(\sigma) \le \ebuffer$.
\end{lemma}
\begin{proof}
  If $r_b(\sigma) = 1$, then \cmark{c:accum-init} gives
  $\eaccum_b(\sigma) = \estart_b(\sigma)$.
  Then \cmark{c:energy-bound} gives
  $\eaccum_b(\sigma) \le \ebuffer$,
  so we have $\estart_b(\sigma) \le \ebuffer$.
\end{proof}

\begin{lemma}
  \label{lem:accum-path}
  Let $\sigma$ be a feasible MILP solution.
  For any CFG path $\pi = b_0 \to b_1 \to \cdots \to b_k$ that stays within the
  same region,
  we have
  \[
  \estart_{b_0}(\sigma) + \sum_{i=0}^{k-1} \eblk_{b_i}(\sigma)
  \le \eaccum_{b_k}(\sigma).
  \]
\end{lemma}
\begin{proof}
  By Lemma~\ref{lem:accum-init}, we have
  $\eaccum_{b_0}(\sigma) = \estart_{b_0}(\sigma)$.
  For each edge $(b_i, b_{i+1})$
  with $0 \le i \le k-1$,
  since $\pi$ stays within the same region, we have $r_{b_{i+1}}(\sigma) = 0$,
  so by \cmark{c:accum-prop}, we have
  $\eaccum_{b_i}(\sigma) + \eblk_{b_i}(\sigma) \le \eaccum_{b_{i+1}}(\sigma)$.
  Chaining these inequalities yields the desired bound.
\end{proof}

\begin{lemma}
  \label{lem:region-energy-bound}
  Let $\sigma$ be a feasible MILP solution.
  If $\pi$ reaches program exit without crossing another region boundary, then
  \[
  \estart_{b_0}(\sigma) +
  \sum_{i=0}^{k} \eblk_{b_i}(\sigma)
  \le \ebuffer.
  \]
  Furthermore, if $\pi'$ reaches the next region boundary, then
  we have
  \[
  \estart_{b_0}(\sigma) +
  \sum_{i=0}^{k} \eblk_{b_i}(\sigma)
  + \eend_{b_{k+1}}(\sigma)
  \le \ebuffer.
  \]
\end{lemma}

\begin{proof}
  If $\pi = b_0 \to b_1 \to \cdots \to b_k$ reaches program exit without
  crossing another region boundary, then by Lemma~\ref{lem:accum-path},
  \[
  \estart_{b_0}(\sigma) + \sum_{i=0}^{k-1} \eblk_{b_i}(\sigma)
  \le \eaccum_{b_k}(\sigma).
  \]
  Since $b_k$ is an exit block, \cmark{c:energy-bound-exit} gives
  \[
  \eaccum_{b_k}(\sigma) + \eblk_{b_k}(\sigma) \le \ebuffer.
  \]
  Summing these inequalities yields the first claim.

  If $\pi' = b_0 \to b_1 \to \cdots \to b_k \to b_{k+1}$ reaches the next
  region boundary, then the prefix $b_0 \to b_1 \to \cdots \to b_k$ stays
  within the same region. By Lemma~\ref{lem:accum-path},
  \[
  \estart_{b_0}(\sigma) + \sum_{i=0}^{k-1} \eblk_{b_i}(\sigma)
  \le \eaccum_{b_k}(\sigma).
  \]
  By \cmark{c:energy-bound}, we have
  \[
  \eaccum_{b_k}(\sigma) + \eblk_{b_k}(\sigma) + \eend_{b_{k+1}}(\sigma)
  \le \ebuffer.
  \]
  Summing these inequalities yields the second claim.
\end{proof}

Now we prove the main energy-safety theorem,
which shows that the energy consumed within each region does not exceed the budget $\ebuffer$.

\getkeytheorem{regionEnergyBound}
\begin{proof}
  If $\pi = b_0 \to b_1 \to \cdots \to b_k$ reaches program exit, then
  $r_{b_0}(\sigma) = 1$ and $r_{b_i}(\sigma) = 0$ for all $1 \le i \le k$.
  By Lemma~\ref{lem:energy-equality-cost-semantics},
  \[
  \operatorname{cost}(\pi(\sigma)) =
  \estart_{b_0}(\sigma) +
  \sum_{i=0}^{k} \eblk_{b_i}(\sigma).
  \]
  By the first part of Lemma~\ref{lem:region-energy-bound}, this is at most
  $\ebuffer$.

  If instead $\pi = b_0 \to b_1 \to \cdots \to b_k \to b_{k+1}$ reaches the
  next region boundary, then $r_{b_0}(\sigma) = r_{b_{k+1}}(\sigma) = 1$ and
  $r_{b_i}(\sigma) = 0$ for all $1 \le i \le k$.
  By Lemma~\ref{lem:energy-equality-cost-semantics},
  \[
  \operatorname{cost}(\pi(\sigma)) =
  \estart_{b_0}(\sigma) +
  \sum_{i=0}^{k} \eblk_{b_i}(\sigma) +
  \eend_{b_{k+1}}(\sigma).
  \]
  By the second part of Lemma~\ref{lem:region-energy-bound}, this is at most
  $\ebuffer$.
\end{proof}

\section{\schematic Bugs}
\label{sec:schematic-bug}

We found three problems in \schematic while porting it.  The first is in
the reference implementation; the second and third are in the paper's
loop-handling algorithm (Algorithm~1) and are inherited by the reference
implementation.  They are independent of one another.  We fixed all three
for fair comparison.

\subsection{Order-dependent inner-loop cost}
\label{sec:schematic-bug-order}

\begin{lstlisting}[language=C]
for (i = 0; i < N; i++) {
  for (j = 0; j < M; j++) {
    if (rare_condition()) {
      // expensive
    }
  }
}
\end{lstlisting}

The inner loop fits within the energy budget and is thus left
checkpoint-free.  However, the worst-case per-iteration cost of the inner
loop is not propagated conservatively to the outer loop.  Consequently, the
outer loop assumes the nested computation is cheaper than it is in the
worst case and makes an overly optimistic checkpointing decision.

The bug is order-dependent.  \schematic analyzes inner-loop traces in
descending frequency order.  When the common inner-loop path is analyzed
first, the shared basic blocks on that path become fixed.  Subsequent
analysis of the rare expensive path is then restricted to the still-unfixed
subregion, so the rare path can no longer fully revise the cost summary
exported by the inner loop.  As a result, the outer loop receives an
optimistic inner-loop cost and chooses an overly large checkpoint interval.

\subsection{Loops absent from the profile}
\label{sec:schematic-bug-profile}

\schematic enumerates the paths of each loop from execution traces.  The
paper handles incomplete traces only for code that is \emph{never
executed}: such code is ``analyzed at the end of the algorithm to ensure
complete code coverage'' (Section~III-A).  The reference implementation
does exactly this: every basic block that no trace covered is placed on a
straight-line path through the CFG and analyzed as if it executed once.

This is unsound when the uncovered code is a loop that does execute.  The
situation arises whenever the compiler's loop structure differs from what
the profiler recorded---for example, when an optimization splits one
source loop into two natural loops and the profiler reports only the outer
one.  The inner loop's iterations are then charged once instead of up to
its trip count, and the enclosing region is sized for fewer iterations
than actually run.

We repair the profile before analysis: every natural loop known to the
compiler receives a loop entry, and for loops the profile did not cover we
synthesize header-to-latch paths from the CFG.  Every loop then goes through
the same loop analysis as a profiled loop, so its iteration count is
always charged.

\subsection{Restore cost omitted from the loop iteration budget}
\label{sec:schematic-bug-restore}

Algorithm~1 of the paper decides how many loop iterations may run between
two consecutive back-edge checkpoints:
\[
  \mathit{num\_it} \leftarrow \left\lfloor E_B / E_{\mathit{loop}} \right\rfloor
  \qquad\text{(Algorithm~1, line~6)}
\]
where $E_{\mathit{loop}}$ is the cost of one iteration.\footnote{The
paper defines $E_{\mathit{loop}} \leftarrow \mathit{header}.E_{\mathit{left}}
- \mathit{latch}.E_{\mathit{left}}$ (line~5), and $E_{\mathit{left}}$ as the
energy remaining \emph{after} a block executes, so this difference omits the
header block's own cost.  The reference implementation computes the iteration
cost from $E_{\mathit{to\_leave}}$ instead, which includes the header.  This
discrepancy is orthogonal to the omission discussed here.}  This budget
charges neither the save at the checkpoint that ends the interval nor the
restore at the checkpoint that begins it.  The paper's own
reachable-checkpoint-graph construction charges both: an edge's cost
comprises ``the energy to restore the volatile data,'' the energy to execute
the path, and the energy to save volatile data, and the edge exists only if
that cost fits in $E_B$ (Section~III-A).  The reference implementation subtracts the save
cost from the loop budget but still omits the restore cost.

The omission is harmless while one iteration is well below $E_B$.  It
produces a wrong decision when
\[
  E_B - E_{\mathit{restore}} - E_{\mathit{save}}
  \;<\; E_{\mathit{loop}}
  \;\le\; E_B - E_{\mathit{save}} .
\]
In this range the paper's formula yields $\mathit{num\_it} \ge 1$ and the
reference implementation's formula (which subtracts one) yields
$\mathit{num\_it} = 0$; both mean ``at least one iteration fits; checkpoint
on the back-edge,'' and neither places a checkpoint inside the loop body.
The region that results---restore, one full iteration, save---does not fit
in $E_B$.

We subtract the restore cost from the iteration budget, matching the
paper's RCG cost model.  When the corrected budget admits no full iteration,
a checkpoint is required inside the loop body.  Algorithm~1 has no such
case: its Step~1 analyzes the body before $\mathit{num\_it}$ is known, and
places body checkpoints only where a single RCG segment exceeds $E_B$.  We
therefore re-analyze the loop's complete header-to-latch path through the
RCG whenever the corrected budget admits no full iteration.

\section{Additional Evaluation Results}
\label{sec:appendix-evaluation}

\Cref{tab:milp-compile-stats-appendix} shows MILP compilation statistics
for all three capacitor sizes (\qty{5}{\micro\farad}, \qty{10}{\micro\farad},
and \qty{50}{\micro\farad}).
\Cref{fig:execution-time-5,fig:execution-time-10,fig:execution-time-50} show
normalized execution time, and
\cref{fig:runtime-region-boundary-calls-5,fig:runtime-region-boundary-calls-10,fig:runtime-region-boundary-calls-50}
show region boundary hits,
at \qty{5}{\micro\farad}, \qty{10}{\micro\farad}, and
\qty{50}{\micro\farad}, respectively.
\Cref{tab:intermittent-decomposition-appendix} breaks down execution
time under intermittent power into execution, checkpointing, and waiting.
\begin{table}[t]
\centering
\caption{MILP compilation statistics of \tool for all three capacitor sizes.
Blocks and edges refer to those in the summarized CFG.
The \qty{10}{\micro\farad} rows repeat \cref{tab:milp-compile-stats}.
All solves reached a guaranteed optimal solution.}
\label{tab:milp-compile-stats-appendix}
\begin{adjustbox}{max width=\linewidth}
\begin{tabular}{lrrrrr}
\toprule
Bench. & Blocks & Edges & Vars. & Constr. & Sol. Time (s) \\
\midrule
\multicolumn{1}{l}{\textit{5\,\textmu{}F}} \\
ar             & \num{8} & \num{9} & \num{533} & \num{1986} & \num{0.026} \\
aes            & \num{9} & \num{10} & \num{2475} & \num{10033} & \num{0.031} \\
bitcount       & \num{36} & \num{51} & \num{5750} & \num{25057} & \num{0.084} \\
chacha20       & \num{4} & \num{4} & \num{495} & \num{1907} & \num{0.004} \\
crc            & \num{8} & \num{9} & \num{157} & \num{573} & \num{0.003} \\
cuckoo\_filter & \num{9} & \num{11} & \num{576} & \num{2327} & \num{0.013} \\
dijkstra       & \num{9} & \num{12} & \num{975} & \num{3937} & \num{0.022} \\
poly1305       & \num{4} & \num{4} & \num{371} & \num{1400} & \num{0.004} \\
qsort          & \num{12} & \num{15} & \num{324} & \num{1264} & \num{0.007} \\
rsa            & \num{26} & \num{33} & \num{4980} & \num{20658} & \num{0.088} \\
sensor\_fusion & \num{174} & \num{263} & \num{24449} & \num{108470} & \num{77.378} \\
sha256         & \num{23} & \num{30} & \num{5827} & \num{23964} & \num{0.182} \\
stringsearch   & \num{18} & \num{24} & \num{955} & \num{3869} & \num{0.238} \\
\midrule
\multicolumn{1}{l}{\textit{10\,\textmu{}F}} \\
ar             & \num{8} & \num{9} & \num{533} & \num{1986} & \num{0.025} \\
aes            & \num{9} & \num{10} & \num{2475} & \num{10033} & \num{0.027} \\
bitcount       & \num{36} & \num{51} & \num{5750} & \num{25057} & \num{0.086} \\
chacha20       & \num{4} & \num{4} & \num{495} & \num{1907} & \num{0.004} \\
crc            & \num{6} & \num{6} & \num{117} & \num{414} & \num{0.003} \\
cuckoo\_filter & \num{7} & \num{8} & \num{432} & \num{1710} & \num{0.005} \\
dijkstra       & \num{9} & \num{12} & \num{975} & \num{3937} & \num{0.021} \\
poly1305       & \num{4} & \num{4} & \num{559} & \num{2135} & \num{0.004} \\
qsort          & \num{12} & \num{15} & \num{324} & \num{1264} & \num{0.006} \\
rsa            & \num{26} & \num{33} & \num{4980} & \num{20658} & \num{0.105} \\
sensor\_fusion & \num{4} & \num{4} & \num{431} & \num{1696} & \num{0.004} \\
sha256         & \num{14} & \num{18} & \num{4083} & \num{16818} & \num{0.079} \\
stringsearch   & \num{18} & \num{24} & \num{955} & \num{3869} & \num{0.225} \\
\midrule
\multicolumn{1}{l}{\textit{50\,\textmu{}F}} \\
ar             & \num{6} & \num{6} & \num{429} & \num{1568} & \num{0.018} \\
aes            & \num{9} & \num{10} & \num{2475} & \num{10033} & \num{0.038} \\
bitcount       & \num{36} & \num{51} & \num{5750} & \num{25057} & \num{0.088} \\
chacha20       & \num{4} & \num{4} & \num{495} & \num{1907} & \num{0.005} \\
crc            & \num{4} & \num{3} & \num{84} & \num{278} & \num{0.001} \\
cuckoo\_filter & \num{5} & \num{5} & \num{222} & \num{833} & \num{0.009} \\
dijkstra       & \num{7} & \num{9} & \num{684} & \num{2734} & \num{0.022} \\
poly1305       & \num{4} & \num{4} & \num{559} & \num{2135} & \num{0.005} \\
qsort          & \num{6} & \num{6} & \num{130} & \num{470} & \num{0.001} \\
rsa            & \num{22} & \num{27} & \num{3979} & \num{16305} & \num{0.373} \\
sensor\_fusion & \num{4} & \num{4} & \num{431} & \num{1696} & \num{0.003} \\
sha256         & \num{8} & \num{9} & \num{1995} & \num{7770} & \num{0.028} \\
stringsearch   & \num{16} & \num{21} & \num{778} & \num{3120} & \num{0.041} \\
\bottomrule
\end{tabular}
\end{adjustbox}
\end{table}

\begin{figure*}[t]
    \centering
    \includegraphics[width=\textwidth]{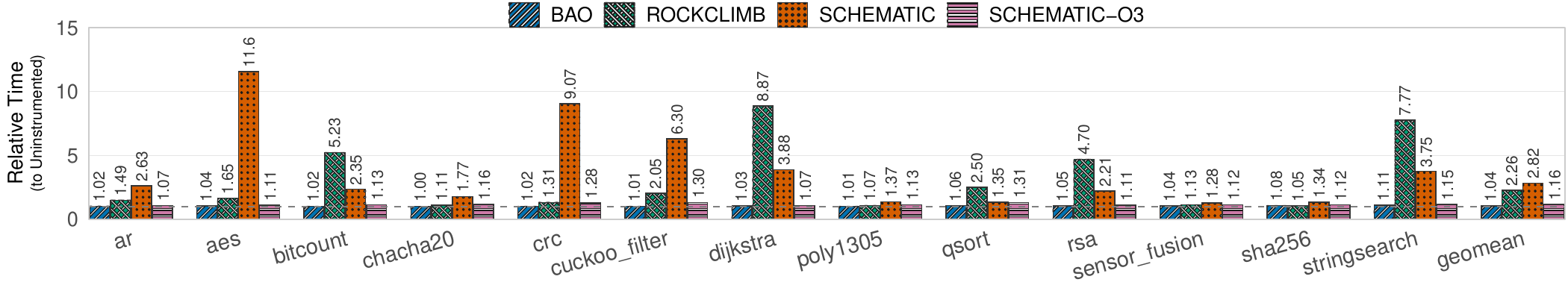}
    \caption{Normalized execution time with 5\,\textmu{}F capacitor.}
    \Description{Bar chart comparing normalized execution time across benchmarks with a 5 microfarad capacitor.}
    \label{fig:execution-time-5}
\end{figure*}

\begin{figure*}[t]
    \centering
    \includegraphics[width=\textwidth]{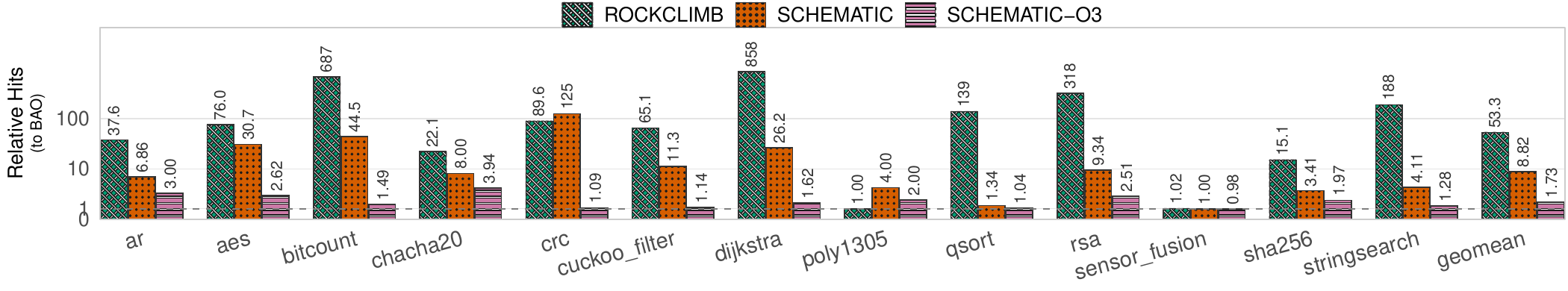}
    \caption{Normalized runtime region boundary hits with 5\,\textmu{}F capacitor.}
    \Description{Bar chart comparing normalized runtime region boundary hits across benchmarks with a 5 microfarad capacitor.}
    \label{fig:runtime-region-boundary-calls-5}
\end{figure*}

\begin{figure*}[t]
    \centering
    \includegraphics[width=\textwidth]{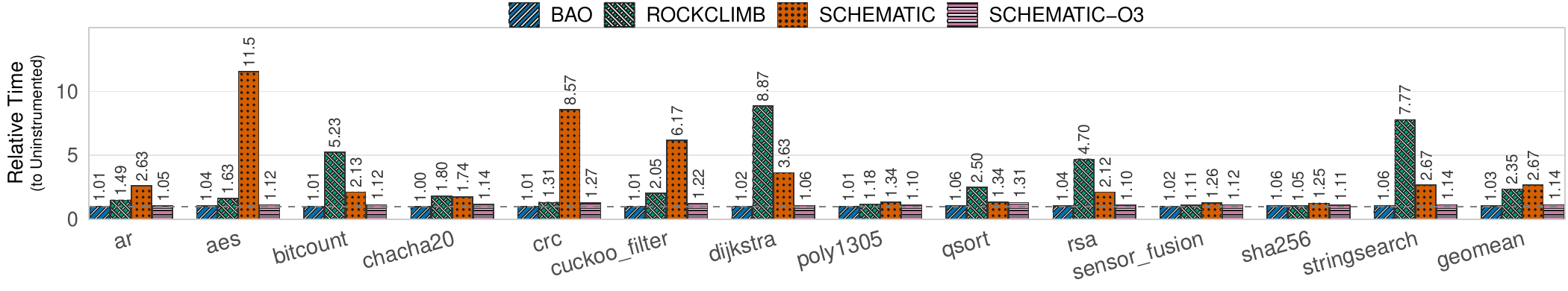}
    \caption{Normalized execution time with 10\,\textmu{}F capacitor.}
    \Description{Bar chart comparing normalized execution time across benchmarks with a 10 microfarad capacitor.}
    \label{fig:execution-time-10}
\end{figure*}

\begin{figure*}[t]
    \centering
    \includegraphics[width=\textwidth]{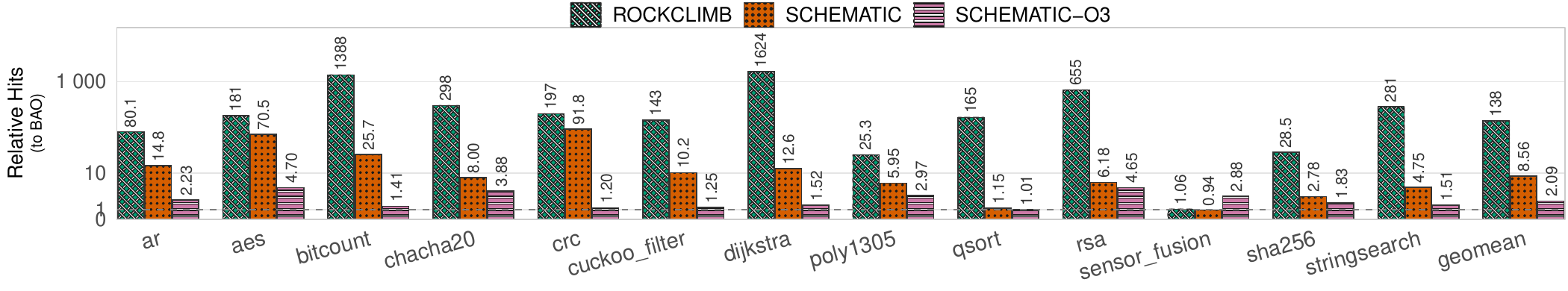}
    \caption{Normalized runtime region boundary hits with 10\,\textmu{}F capacitor.}
    \Description{Bar chart comparing normalized runtime region boundary hits across benchmarks with a 10 microfarad capacitor.}
    \label{fig:runtime-region-boundary-calls-10}
\end{figure*}

\begin{figure*}[t]
    \centering
    \includegraphics[width=\textwidth]{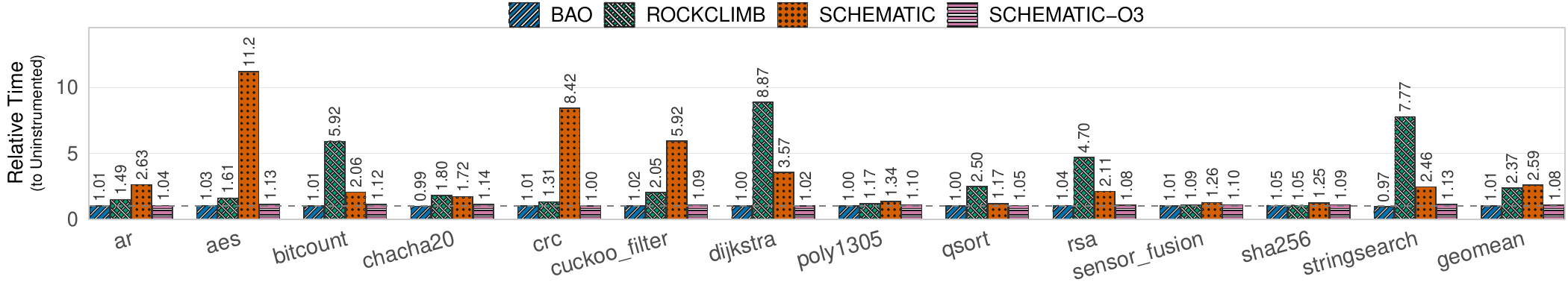}
    \caption{Normalized execution time with 50\,\textmu{}F capacitor.}
    \Description{Bar chart comparing normalized execution time across benchmarks with a 50 microfarad capacitor.}
    \label{fig:execution-time-50}
\end{figure*}

\begin{figure*}[t]
    \centering
    \includegraphics[width=\textwidth]{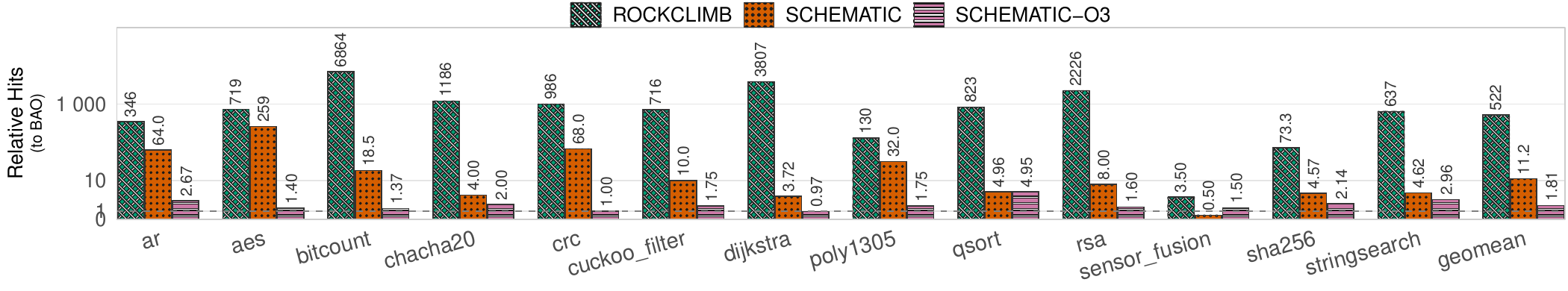}
    \caption{Normalized runtime region boundary hits with 50\,\textmu{}F capacitor.}
    \Description{Bar chart comparing normalized runtime region boundary hits across benchmarks with a 50 microfarad capacitor.}
    \label{fig:runtime-region-boundary-calls-50}
\end{figure*}

\begin{table*}[t]
    \centering
    \caption{Per-benchmark decomposition of execution time under
    intermittent power, as the mean share (\%) of each run's total spent in
    $T_{\mathrm{exec}}$ / $T_{\mathrm{ckpt}}$ / $T_{\mathrm{wait}}$.
    Boot and restore after power failures take the remainder, below 0.1\%
    everywhere.}
    \label{tab:intermittent-decomposition-appendix}
\begin{tabular*}{\textwidth}{l@{\extracolsep{\fill}}cccccccccccc}
\toprule
Benchmark & \multicolumn{3}{c}{\tool} & \multicolumn{3}{c}{\rockclimb} & \multicolumn{3}{c}{\schematic} & \multicolumn{3}{c}{\schematicO} \\ \cmidrule{2-4} \cmidrule{5-7} \cmidrule{8-10} \cmidrule{11-13}
~ & $T_{\mathrm{exec}}$ & $T_{\mathrm{ckpt}}$ & $T_{\mathrm{wait}}$ & $T_{\mathrm{exec}}$ & $T_{\mathrm{ckpt}}$ & $T_{\mathrm{wait}}$ & $T_{\mathrm{exec}}$ & $T_{\mathrm{ckpt}}$ & $T_{\mathrm{wait}}$ & $T_{\mathrm{exec}}$ & $T_{\mathrm{ckpt}}$ & $T_{\mathrm{wait}}$ \\
\midrule
aes & 5.9 & 0.4 & 93.7 & 0.4 & 1.6 & 98.0 & 0.2 & 1.9 & 98.0 & 2.0 & 0.5 & 97.6 \\
crc & 14.1 & 0.4 & 85.4 & 0.5 & 1.9 & 97.7 & 0.2 & 1.7 & 98.2 & 9.5 & 2.8 & 87.6 \\
rsa & 7.9 & 0.7 & 91.5 & 0.1 & 2.2 & 97.8 & 0.9 & 1.4 & 97.6 & 2.7 & 0.8 & 96.5 \\
dijkstra & 8.2 & 0.7 & 91.1 & 0.0 & 2.1 & 97.8 & 0.5 & 1.8 & 97.7 & 7.3 & 1.2 & 91.5 \\
qsort & 7.9 & 0.4 & 91.6 & 0.2 & 2.1 & 97.7 & 3.3 & 1.9 & 94.8 & 5.0 & 2.3 & 92.8 \\
activity\_recognition & 8.3 & 0.6 & 91.1 & 0.5 & 1.7 & 97.9 & 0.8 & 1.9 & 97.3 & 6.5 & 1.2 & 92.4 \\
bitcount & 9.8 & 0.6 & 89.6 & 0.1 & 2.2 & 97.8 & 1.1 & 2.2 & 96.8 & 8.6 & 0.8 & 90.6 \\
chacha20 & 11.6 & 0.2 & 88.2 & 1.6 & 1.4 & 97.0 & 1.9 & 1.7 & 96.4 & 3.7 & 0.8 & 95.6 \\
sensor\_fusion & 8.0 & 0.7 & 91.2 & 6.8 & 1.3 & 91.9 & 5.0 & 1.9 & 93.1 & 5.6 & 1.9 & 92.7 \\
poly1305 & 3.2 & 0.1 & 96.7 & 1.0 & 0.7 & 98.3 & 1.5 & 0.7 & 97.9 & 2.0 & 0.3 & 97.6 \\
cuckoo\_filter & 8.7 & 0.8 & 90.5 & 0.2 & 2.0 & 97.7 & 0.3 & 1.9 & 97.8 & 3.3 & 1.3 & 95.4 \\
sha256 & 7.0 & 0.8 & 92.2 & 1.2 & 0.8 & 98.1 & 2.7 & 1.0 & 96.3 & 3.5 & 0.5 & 95.9 \\
stringsearch & 6.1 & 1.9 & 92.2 & 0.0 & 2.2 & 97.8 & 0.6 & 2.1 & 97.3 & 3.9 & 1.7 & 94.4 \\
\bottomrule
\end{tabular*}
\end{table*}

\clearpage

\section{Microbenchmarks for Energy Model Validation}
\label{sec:appendix-microbenchmarks}

\Cref{lst:microbenchmarks} shows the four synthetic microbenchmarks,
each repeating a small basic block with a different addressing mode.
\Cref{tab:microbenchmark-error} reports the energy estimation error
on each.
\begin{listing}[!h]
\caption{Synthetic Microbenchmarks (Excerpted).}
\begin{lstlisting}[language=C, label={lst:microbenchmarks}]
void bench_register_only(void) {
  // .rept REPT
  //   mov.w %[src], %[dst]
  //   add.w %[step], %[dst]
  //   cmp.w %[expected], %[dst]
  //   jnz 1f
  // .endr
}

void bench_load_compute(void) {
  // .rept REPT
  //   mov.w 0(%[base_src]), %[dst]
  //   add.w #3, %[dst]
  //   cmp.w #0x1114, %[dst]
  //   jnz 1f
  // .endr
}

void bench_load_compute_store(void) {
  // .rept REPT
  //   mov.w 0(%[base_src]), %[tmp]
  //   add.w %[addend], %[tmp]
  //   mov.w %[tmp], 0(%[base_dst])
  //   cmp.w #0x3333, %[tmp]
  // .endr
}

void bench_pointer_walk(void) {
  // .rept REPT
  //   mov.w @%[stream_a]+, %[tmp0]
  //   mov.w @%[stream_b]+, %[tmp1]
  //   add.w %[tmp0], %[acc]
  //   cmp.w #0x4000, %[acc]
  // .endr
}
\end{lstlisting}
\end{listing}

\begin{table}[h]
\centering
\caption{Estimation error on synthetic microbenchmarks.}
\label{tab:microbenchmark-error}
\small
\begin{tabular*}{\linewidth}{@{}llr@{}}
\toprule
Microbenchmark & Addressing mode & Error \\
\midrule
reg.~only  & Register-only arithmetic       & $-0.46$\% \\
ld+comp    & Indexed load, compute          & $+0.07$\% \\
ld+comp+st & Indexed load, compute, store   & $-0.76$\% \\
ptr~walk   & Auto-increment pointer walk    & $-1.19$\% \\
\midrule
\multicolumn{2}{@{}l}{Mean absolute error} & $0.62$\% \\
\bottomrule
\end{tabular*}
\end{table}

}

\end{document}